\documentclass[12pt]{article}
\usepackage{jheppub} % for details on the use of the package, please                    % see the JHEP-author-manual
\usepackage{amsmath,braket,physics,amsfonts,amsthm,amssymb,mathtools}
\graphicspath{{figures/} }
\usepackage{graphicx,hyperref,xcolor}
\usepackage{enumitem}
\usepackage{mathrsfs}
\hypersetup{
    bookmarks=true,         % show bookmarks bar?
    unicode=false,          % non-Latin characters in Acrobat’s bookmarks
    pdftoolbar=true,        % show Acrobat’s toolbar?
    pdfmenubar=true,        % show Acrobat’s menu?
    linktocpage=true,       % hyperlink by page number in contents
    pdffitwindow=false,     % window fit to page when opened
    pdfstartview={FitH},    % fits the width of the page to the window
    pdfnewwindow=true,      % links in new PDF window
    colorlinks=true,       % false: boxed links; true: colored links
    linkcolor=blue,          % color of internal links (change box color with linkbordercolor)
    citecolor=blue,        % color of links to bibliography
    filecolor=blue,      % color of file links
    urlcolor=blue           % color of external links
}
\numberwithin{equation}{section}									% equation numbering by section

\usepackage{comment}
\usepackage{xcolor}

\newcommand{\be}{\begin{equation}}
\newcommand{\ba}{\begin{eqnarray}}
\newcommand{\ea}{\end{eqnarray}}
\newcommand{\ee}{\end{equation}}

\newcommand{\bea}{\begin{eqnarray}}
\newcommand{\eea}{\end{eqnarray}}
\newcommand{\bes}{\begin{equation*}}
\newcommand{\beas}{\begin{eqnarray*}}
\newcommand{\eeas}{\end{eqnarray*}}
\newcommand{\bas}{\begin{array*}}
\newcommand{\eas}{\end{array*}}
\newcommand{\ees}{\end{equation*}}

\newtheorem{theorem}{Theorem}[section]
\newtheorem{corollary}{Corollary}[section]
\newtheorem{conjecture}{Conjecture}[section]
\newtheorem{definition}{Definition}[section]
\newtheorem{proposition}{Proposition}[section]
\newtheorem{remark}{Remark}[section]
\newtheorem{lemma}{Lemma}[section]

\usepackage{tikz}
\usetikzlibrary{arrows,arrows.meta,intersections, calc,positioning,decorations.pathreplacing,decorations.pathmorphing,shapes}
\usetikzlibrary{patterns}
\usetikzlibrary{decorations.markings}
\usetikzlibrary{knots}

\title{Non-isometry, State Dependence and Holography}
\author{Stefano Antonini$^{1}$,Vijay Balasubramanian$^{2,3,4}$, Ning Bao$^{5,6}$, ChunJun Cao$^{7}$ and Wissam Chemissany$^2$}
\affiliation[1]{Berkeley Center for Theoretical Physics and Department of Physics, University of California, Berkeley, California 94720, U.S.A}
\affiliation[2]{David Rittenhouse Laboratory, University of Pennsylvania, Philadelphia, PA 19104, USA}
\affiliation[3]
{Rudolf Peierls Centre for Theoretical Physics, University of Oxford, UK} 
\affiliation[4]
{Theoretische Natuurkunde, Vrije Universiteit Brussel and International Solvay Institutes, Pleinlaan 2, B-1050 Brussels, Belgium}
\affiliation[5]{Computational Science Initiative, Brookhaven National Lab, Upton, NY 11973, USA}
\affiliation[6]{Department of Physics, Northeastern University, Boston, MA 02115, USA}
\affiliation[7]{Department of Physics, Virginia Tech, Blacksburg, VA, 24061, USA}

\emailAdd{santonini@berkeley.edu, vijay@physics.upenn.edu,\\
ningbao75@gmail.com, cjcao@vt.edu, chemissany.wissam@sas.upenn.edu}

\abstract{We establish an equivalence between non-isometry of quantum codes and state dependence of operator reconstruction, and discuss implications of this equivalence for holographic duality. Specifically, we define quantitative measures of non-isometry and state dependence and describe bounds relating these quantities. In the context of holography we show that, assuming known gravitational path integral results for  overlaps between semiclassical states, non-isometric bulk-to-boundary maps with a trivial kernel are approximately isometric and bulk reconstruction approximately state-independent. In contrast, non-isometric maps with a non-empty kernel always lead to state-dependent reconstruction. We also show that if a global bulk-to-boundary map is non-isometric, then there exists a region in the bulk which is causally disconnected from the boundary. Finally, we conjecture that, under certain physical assumptions for the definition of the Hilbert space of effective field theory in AdS space, the presence of a global horizon implies a non-isometric global bulk-to-boundary map.
}

\begin{document}

\maketitle

\section{Introduction}
%\sa{We should clarify our definition of state-dependence and how it relates to the definition of state specific in akers-penington, akers et al., etc.}

Black holes are believed to have an entropy proportional to their horizon area \cite{Bekenstein:1972tm,Bekenstein:1973ur,Bekenstein:1974ax,Hawking:1974rv,Hawking:1975vcx}. If this is correct, the local degrees of freedom described by effective field theory (EFT) in the region behind the horizon cannot all be independent.  Similarly, in the holographic AdS/CFT correspondence \cite{tHooft:1993dmi,Susskind:1994vu,Maldacena:1997re,Witten:1998qj,Gubser:1998bc,Aharony:1999ti}, the degrees of freedom of a higher dimensional quantum gravity emerge from a lower dimensional local quantum field theory on the spacetime boundary.  In this case too, a surface area's worth of degrees of freedom is apparently sufficient to encode the local dynamics of the gravitating interior.

In some circumstances, non-perturbative effects in quantum gravity are known to achieve this winnowing of the overcomplete local degrees of freedom in the low energy effective theory.  For example, wormholes can introduce small overlaps of $O(\exp(-\gamma/G))$ between apparently orthogonal states of low energy effective field theories. Such effects have played a prominent role in recent analyses of information recovery from black holes \cite{penington2020entanglement,almheiri2019entropy,Penington:2019kki,almheiri2020replica} and of the microcanonical origin of the entropy of black holes in quantum gravity \cite{Balasubramanian:2022gmo,Balasubramanian:2022lnw}.\footnote{Similar effects are also central in the resolution of the factorisation puzzle in AdS/CFT settings \cite{Guica:2015zpf,Harlow:2015lma,Harlow:2018tqv,Penington:2023dql,Boruch:2024kvv}.} Indeed, even though  non-perturbatively induced overlaps are exponentially small, they can result in an exponential compression of an EFT Hilbert space  \cite{william1984extensions, candes2006near, chao2017overlapping, buhrman2001quantum} while nevertheless approximately preserving many features of the effective theory such as spatial locality \cite{papadodimas2013infalling,chao2017overlapping, cao2023overlapping}. Thus, the apparently orthogonal low energy states are a non-isometric representation of the true underlying quantum gravity Hilbert space in the sense that they do not preserve the inner product.  This sort of non-isometry provides a concrete quantum information-theoretic mechanism to resolve the apparent mismatch between the counting of local EFT degrees of freedom and the entropies predicted by gravity, and has recently found application to the reconstruction of the black hole interior in the late stages of evaporation \cite{Akers:2021fut,akers2022black,balasubramanian2023quantum}.

%A profound observation about black holes is that the apparent degrees of freedom enclosed by the horizon would naively scale as the volume of the black hole even though its entropy is proportional to the horizon area. This “holographic principle” is subsequently extended to other spacetimes where a similar higher dimensional bulk can be encoded in a lower dimensional manifold[], of which the AdS/CFT correspondence is a concrete realization[]. However, the precise manner local degrees of freedom in higher dimensions emerge from the boundary or horizon degrees of freedom remains an unsolved problem. 

%It is long believed that the breakdown of exact locality in quantum gravity would play a pivotal role in the emergence of bulk locality. In particular, these non-perturbative quantum gravity corrections induce small overlaps of $O(\exp(-\gamma/G))$ between otherwise complete orthogonal basis states in the effective field theory description and turn them into an overcomplete one[e.g. Euclidean wormhole]. Despite the corrections being exponentially small, the overlaps allow for  an exponentially compression of the EFT Hilbert space[Johnson Lindenstrauss, compressed sensing, OQ, etc] while approximately preserving many salient features of the effective theory [PR,Chao, dS/MERA etc]. Such non-isometric compressions therefore provide a concrete quantum information-theoretic mechanism to resolve the apparent mismatch in the local degrees of freedom counting and has recently found applications in the reconstruction of BH interior and the related information problem[Vijay, Akers].

However, in circumstances where the AdS/CFT bulk-to-boundary map is non-isometric, some familiar properties of the holographic dictionary, such as state-independent unitary reconstruction in the CFT of operators in AdS, are expected to break down \cite{Akers:2020pmf,Akers:2021fut,akers2022black}. In fact, (approximate) state-independent reconstruction of all bulk operators would necessarily require the bulk-to-boundary map be (nearly) isometric, which cannot be true for local excitations in the black hole interior if the horizon area measures the entropy \cite{Akers:2020pmf,Akers:2021fut,Cao_2021,DeWolfe:2023iuq,DeWolfe:2023jfm,Steinberg_2023,Bueller:2024zvz}. We also expect a connection between non-isometry and state-dependence for holographic subregion reconstruction in the context of  holographic Quantum Error Correcting (QEC) codes, when the encoding map for a bulk subregion is only approximately isometric, as occurs away from the classical limit \cite{Engelhardt:2014gca,faulkner2020holographic}.  Non-isometry and state-dependence have also been related in holographic toy models and more generally in tensor networks \cite{Hayden_2016,Hayden:2018khn,Cao_2021,charlesnote,cao2023overlapping, Bueller:2024zvz,Steinberg_2023} %{\color{red} (ABSC, Noniso,  and other related works)}.

Given these examples, our goal in this paper is to precisely quantify the relationship between non-isometry and state-dependence.
%Our work here seeks to bridge this gap and formulate this connection rigorously. 
In Sec.~\ref{sec:QI}, we define measures for non-isometry and state-dependence by averaging over states and reconstructed operators. We identify an equivalence relation between these concepts---the more non-isometric a map is, the more state-dependence one needs for operator reconstruction, and vice versa. These general quantum information results hold for any linear maps over finite dimensional Hilbert spaces. Then in Sec.~\ref{sec:gravity}, we study consequences of this equivalence for quantum gravity. To do so we have to go beyond the average case measures of Sec.~\ref{sec:QI}. We show that if we insist that the overlap between any two semiclassical states must be consistent with gravitational path integral calculations and that their inner product is approximately preserved by a kernel-less bulk-to-boundary map, then the map is also approximately isometric and the reconstruction of bulk operators is approximately state-independent. On the other hand, non-isometric maps with kernels require state-dependent reconstruction even if we restrict  attention to a subset of semi-classical states and operators \cite{akers2022black}. We further show that if the global bulk-to-boundary map is non-isometric, then there exists a region of the bulk which is causally disconnected from the boundary. This can be achieved by the existence of either a global future and past horizons, or a disconnected patch of spacetime without a boundary (namely, a closed universe). Finally, under physically-motivated assumptions about the bulk EFT degrees of freedom giving rise to the bulk EFT Hilbert space, we conjecture that in the presence of a global horizon the global bulk-to-boundary map must be non-isometric. Finally, in Sec.~\ref{sec:discussion} we discuss our results and future directions.

\section{Non-isometry and state-dependence}
\label{sec:QI} 

Let $A,B$ be Hilbert spaces. A linear map $V:A\rightarrow B$ is isometric if it preserves the inner product, i.e., $\langle\psi|V^{\dagger}V|\phi\rangle =\langle\psi|\phi\rangle$ for all states $|\psi\rangle,|\phi\rangle\in A$. Such maps are commonly used as quantum encoding maps as they map states in a logical Hilbert space to corresponding representations in a physical Hilbert space, also known as codewords. We refer to any map that is not isometric as {\it non-isometric}. In this work, we only consider maps between finite dimensional Hilbert spaces. Concretely, we can let $V$ be the global holographic bulk to boundary map in AdS/CFT, and take $A$ to be the Hilbert space of effective degrees of freedom, and $B$ to be the fundamental Hilbert space.

In contrast to non-isometry, state-independence is a less familiar concept in the quantum coding literature. The term also seems to have been given various, sometimes inequivalent, definitions  in the literature \cite{Papadodimas_2015,Papadodimas_2016,Papadodimas_2013,papadodimas2013infalling,papadodimas2014black,Marolf_2016,Hayden:2018khn,Akers:2021fut,Akers:2020pmf,Cao_2021,Yoshida_2019}.  So we will start by defining the notion of state-dependence that we will use in the present paper. From \cite{Akers:2021fut,akers2022black}, %{\color{red} (Akers-Penington)}, 

\begin{definition}[State-specific reconstruction]
    For a fixed $|\psi\rangle\in A$, a unitary operator $O$ on A admits a state-specific reconstruction if there exists a unitary $\tilde{O}$ on $B$ such that $\tilde{O}V|\psi\rangle=VO|\psi\rangle$. 
\end{definition}
%{\color{blue} VB: Do we really need the separate definition above?  It's a bit clunky.  I think one would generally understand the term ``state-specific reconstruction'' to mean that $O$ varies between states. }

\begin{definition}[State-independent reconstruction]
\label{def:stateindependence}
    A unitary operator $O$ on $A$ admits a state-independent reconstruction if there exists a fixed unitary $\tilde{O}$ on $B$ such that $\tilde{O} V |\psi\rangle = V O |\psi\rangle$ for all $|\psi\rangle \in A$.
%   it is a valid state-specific reconstruction of $O$ for all $|\psi\rangle\in A$.
\end{definition}
%{\color{blue} VB: I restored the equation in the definition.}
%In other words, let $S_{\tilde{O}}(|\psi\rangle)$ be the set of all valid state-specific reconstructions $\tilde{O}$ of $O$ with respect to $|\psi\rangle$, this is to say that the reconstruction is state independent if $\bigcap_{|\psi\rangle\in A} S_{\tilde{O}}(|\psi\rangle )\ne \emptyset$.

When no fixed reconstruction $\tilde{O}$ exists, namely  $\tilde{O}V|\psi\rangle=VO|\psi\rangle$ only holds for some states $|\psi\rangle$, we say the reconstruction is {\it state-dependent}. Similarly, if a fixed $\tilde{O}$ suffices for all states in a  subspace of $A$, but doesn't for states in another subspace, then we can say the reconstruction is {\it subspace-dependent}. 
Note that a state-specific reconstruction here need not be state-dependent. This definition allows us to introduce a computable measure for state dependence. We will further refine this notion at the end of this Section.

Next, we will make precise the relation between state-dependence and non-isometry.

\begin{theorem}
    A linear map $V$ is isometric if and only if all unitary logical operators $O$ can be reconstructed in a state-independent manner.
    \label{thm:exactrecon}
\end{theorem}

\begin{proof}
Let $V:A\rightarrow B$ where we refer to $A$ as the logical Hilbert space and $B$ as the physical Hilbert space. Here by logical operator we mean operators that act on $A$, the input space.
    If $V$ is isometric, then it is well-known that any logical operators can be unitarily represented in the physical space. It is easy to give a construction for any logical operator $\tilde{O}$ by writing down the singular value decomposition of $V$. Hence the reconstruction must be state-independent. Therefore state-dependence would imply non-isometry. 

   For the other direction, let us suppose that $V$ is non-isometric, and show that there exists a unitary logical operator that does not admit a unitary state-independent reconstruction.
   Recall that $O$ admits one such reconstruction if there is some $\tilde{O}$ such that $\tilde{O}V|\psi\rangle=VO|\psi\rangle,\forall |\psi\rangle$. Because $V$ is non-isometric, there must exist states $|\psi\rangle,|\phi\rangle, |||\phi\rangle||=|||\psi\rangle||$ such that the norm induced by the inner product is not preserved under the map, i.e., $||V|\psi\rangle||\ne ||V|\phi\rangle||$. Let $O$ be a unitary such that $|\psi\rangle = O|\phi\rangle$, then $||\tilde{O}V|\phi\rangle||=||V|\phi\rangle||\ne ||V|\psi\rangle||=||VO|\phi\rangle||$. Since a unitary $\tilde{O}$ is norm preserving, there is no reconstruction $\tilde{O}$ for $O$ that can possibly map one vector to another with a different norm.\footnote{Notice that for this operator $O$ there simply exists no unitary reconstruction, not even state-specific.}
   \end{proof}

%In other words, non-isometry is equivalent to state-dependence when we demand exact reconstruction to be satisfied. 

%The more detailed proof is given in App.~\ref{app:proof_thm2.1}, [proof lemma 3.2 note] however it is instructive to build up some intuition from a sketch. 

%The forward direction is straightforward, as it follows from our usual expectation in quantum error correction. By restricting ourselves to the code subspace, all logical operators can be reconstructed in a state-independent manner. For the reverse direction, suppose the map is non-isometric, then it is clear some vectors will be stretched or shrunk under this transformation. Suppose a logical operator $O$ swaps two states that are stretched by different amount under $V$, then no unitary $\tilde{O}$ can accomplish the same feat as it preserves the length of the vector. A similar statement can also be found as Theorem H.1 of \cite{}[Akers et al BH] where they showed that if $V$ is proportional to a random projection, then the state-independent reconstruction of all subexponential operators (or just Pauli operators) implies approximate isometry.

In other words, non-isometry is equivalent to state-dependence if we require exact reconstruction of operators.  We would like to generalize this correspondence beyond exact isometry and exact state-independence.  Namely, if operator reconstruction is approximately state-independent, does it similarly imply $V$ is approximately isometric and vice versa?  We will show that the answer is ``yes''. To make this precise, we first need to define  quantitative measures of non-isometry and state-dependence.

%\CC{set a consistent notation for the norms}
\subsection{Quantifying non-isometricity}
 
\begin{definition}
\label{def:nonisometry_measure}
Let $V:A\rightarrow B$ be a linear map between Hilbert spaces. The degree of non-isometricity of $V$ is
\begin{equation}
    \mathcal{D}(V) = \int d\mu(|\psi\rangle)d\nu(|\phi\rangle)|\langle\psi|V^{\dagger}V|\phi\rangle -\langle\psi|\phi\rangle|^2
    \label{eqn:nonisomeasure}
\end{equation}
 where $d\mu$ and $d\nu$ indicate the Haar measure.
%normalized uniform measure.
\end{definition}
By the Haar measure on states we mean that we sum over $|\psi\rangle = U |0\rangle$ where $U$ are unitaries taken from the Haar measure and $|0\rangle$ is a reference state.  $ \mathcal{D}(V)$ vanishes by construction if $V$ is isometric, and is manifestly positive.  Likewise it follows from the definition that if $W$ is an isometry, then $ \mathcal{D}(WV) =  \mathcal{D}(V)$. Thus we have:
%Intuitively, because isometries preserve inner products, $\mathcal{D}(V)$ quantifies how much of that is altered under $V$. We will take $||\cdot||$ to be the vector 2-norm. This is a valid quantifier of non-isometry because
\begin{proposition}
\label{prop:cDv}
$\mathcal{D}(V)$ satisfies the following conditions
\begin{enumerate}[label=\roman*)]
    \item $\mathcal{D}(V)\geq 0$ and attains equality if and only if $V$ is isometric. 
    \item $\mathcal{D}(V)$ is invariant under isometric composition $\mathcal{D}(WV)=\mathcal{D}(V)$ for isometries $W$.
\end{enumerate}
\end{proposition}
%The first property  $\mathcal{D}(V)\geq 0$ follows from the definition (\ref{eqn:noniso_measure}). Similarly, by continuity, it vanishes if and only if all pairs of inner products are preserved. The second is obvious by definition.

For some maps $V$ it may be convenient to generalize the  non-isometry measure (\ref{eqn:nonisomeasure})  to sum only over specific subsets of states $S\subset A$ rather than the entire space.  
For instance, in quantum gravity, one might want to focus on semi-classical states for which the background geometry is well-defined.  In this section we will derive general results, and will discuss gravity-related implications in the next section.

If $V$ is injective it must have a trivial kernel and $|A|\leq |B|$, where by $|A|$ we mean the dimension of $A$. This means that no elements of $A$ will be mapped to $0$, but $V$ can still be non-isometric by changing the lengths of vectors.  Such \emph{weakly non-isometric} maps can be found in certain constructions of approximate quantum error correcting codes, e.g. \cite{Cao_2021,Hayden_2016} as global maps. In AdS/CFT, they arise both in the context of subregion reconstruction \cite{Akers:2021fut}---which we will not discuss here---and for global maps \cite{brown2020python,Engelhardt:2021mue,Engelhardt:2021qjs}.
We will refer to maps $V$ that are not injective, and hence have a nontrivial kernel, as \emph{strongly non-isometric}. 
If $V$ is surjective, the rank-nullity theorem tells us that $|B|\leq|A|$. Such maps have been studied in \cite{candes2006near,chao2017overlapping,randomproj,randomproj2} and in \cite{akers2022black,Akers:2021fut,cao2023overlapping,balasubramanian2023quantum}. 

Note that any map $V$ that is non-surjective can be made surjective by restricting to the image of $V$. In other words, one can always identify $B'=Im(V)$ and redefine $V':A\rightarrow B'$ where $V'$ is surjective. Since our non-isometry measure is invariant under such a redefinition, we only need to focus on proofs for surjective non-isometric maps for the remainder of this section. However, for the above reason, we clarify that these measures are fully valid for both the strongly and weakly non-isometric maps regardless of their surjectivity. %that will be discussed in the context of holography in Sec. \ref{sec:gravity} regardless of their surjectivity. 
%This rewriting has no impact on the values of the measures we are about to define. We can now divide up the different cases $V$ is non-isometric. 
%\CC{think about this more carefully}

%Evaluating the integrals, one can easily show that (App.~\ref{app:proof_lemma2.1})
\begin{lemma}\label{lemma:non_isometry}
Let $V:A\rightarrow B$ be a linear map. Then
    \begin{equation}
    \mathcal{D}(V) = \Tr[(V^{\dagger}V-I)^2]/|A|^2.
\label{eqn:noniso_measure}
\end{equation}
If we further normalize as $\Tr[V^{\dagger}V]=|A|$, then 
\begin{equation}\label{eqn:non-iso-square}
    \mathcal{D}(V)=\frac{1}{2|A|^3}\sum_{i,j=1}^{|A|}(\lambda_i-\lambda_j)^2
\end{equation}
where $\{\lambda_i\}$ are the eigenvalues of $V^{\dagger}V$, including those that vanish.
\end{lemma}
%\begin{proof}
%See Appendix \ref{app:proof_lemma2.1}.
%\end{proof}
\begin{proof}
Let $V:A\rightarrow B$ and $d=|A|$.
    \begin{align}
    \mathcal{D}(V) &= \int d\mu(|\psi\rangle)d\nu(|\phi\rangle) |\langle\psi|V^{\dagger}V|\phi\rangle -\langle\psi|\phi\rangle|^2\\
    &=\int d\mu(|\psi\rangle)d\nu(|\phi\rangle) \Tr[|\psi\rangle\langle\psi|(\Pi-I)|\phi\rangle\langle\phi|(\Pi-I)]\\
    &=\frac{1}{d^2}\Tr[(\Pi-I)^2],
\end{align}
where we have defined $V^{\dagger}V=\Pi$. If $\Tr[\Pi]=d$ where $d$ is the dimension of the ``logical Hilbert space'', then 
\begin{equation}
    \mathcal{D}(V)=\frac{1}{2d^3} \sum_{i,j}(\lambda_i-\lambda_j)^2.
\end{equation}
We arrive at this by observing that $\Tr[(\Pi\otimes I-I\otimes\Pi)^2]=2d\Tr[(\Pi-I)^2]$. The eigenvalues of $\Pi$ are $\{\lambda_i\}$ where we include the zero eigenvalues for vectors in the null space. 
\end{proof}

We can refine the non-isometry measure above by separating it into  parts that  capture the two different sources of non-isometricity: stretching of vectors and the presence of a kernel in the map. To this end, let $Z$ index basis states  in the kernel, and let $Z^c$ be its complement. Then
\begin{equation}
   \mathcal{D}(V)=\mathcal{D}_{W}(V)+\mathcal{D}_{S}(V)
\end{equation}
where 
\begin{align}
    \mathcal{D}_{W}(V)&=\frac{1}{2|A|^3}\sum_{i,j\in Z^c} (\lambda_i-\lambda_j)^2\\
   \mathcal{D}_{S}(V) &= \frac{1}{|A|^3}\sum_{i\in Z, j\in Z^c}(\lambda_i-\lambda_j)^2 = \frac{\dim(ker(V))}{|A|^3}\Tr[(V^{\dagger}V)^2].
   \label{eq:Ds}
\end{align}
Heuristically, $\mathcal{D}_{S}(V)$ captures \emph{strong non-isometricity} coming from  non-injectivity, as it vanishes when $V$ has a trivial kernel.  Likewise, $\mathcal{D}_{W}(V)$ captures  \emph{weak non-isometricity} which arises if the map stretches vectors without shrinking any to zero length. This is qualitatively similar to the definitions by \cite{Bueller:2024zvz}.
%\footnote{Such separation is of course artificial as one can think of mapping a vector in the kernel as shrinking its length to zero.}.
%\CC{add or leave out relative entropy measure?}

\subsection{Quantifying state dependence}\label{subsec:statedep}
There is more than one way to measure state dependence when one deviates from exact state independence. 
One possible measure for state dependence is to determine how badly exact state-independent reconstruction is broken for logical operators in $A$ for an ``encoding map'' $V$. %\CC{Consider changing notation, have been using different d's.}

%\footnote{For simplicity, here and in the following we refer to this measure as a ``measure of state-dependence''. However, a more precise characterization would be a ``measure of failure of state-independent reconstruction''. The reason is that this measure is non-zero not only for operators which can only be reconstructed state-dependently, but in general for all operators that cannot be reconstructed state-independently. This includes operators that do not admit a unitary reconstruction at all (not even a state-dependent one), such as those considered in the proof of Theorem \ref{thm:exactrecon}.} 
\begin{definition}
    The degree of state-dependence in the reconstruction of an operator $O$ is 
    %$\min_{\tilde{O}\in U(B)}D_V(\tilde{O},O)$ where 
    $$
    \min_{\tilde{O}\in U(B)}D_V(\tilde{O},O) ~~~{\rm where}~~~    
    D_V(\tilde{O},O)=\int d\mu(|\psi\rangle)||\tilde{O}V|\psi\rangle-VO|\psi\rangle||^2$$
    where $\mu(|\psi\rangle)$ is the Haar measure over $A$ and $U(B)$ is the set of unitary operators on $B$. $||\cdot||$ is the standard vector 2-norm.
\end{definition}
The Haar measure is defined below Definition~\ref{def:nonisometry_measure}. $D_V$ vanishes by definition if and only if $O$ admits an exact state-independent reconstruction.
Evaluating the integral, one obtains
    \begin{equation}
    D_V(\tilde{O},O) = \frac{2}{|A|}\Big(\Tr[V^{\dagger}V] - Re\left\{\Tr[V^{\dagger}\tilde{O}VO^{\dagger}]\right\}\Big)
\end{equation}
To quantify how far we are from reconstructing the entire logical algebra in a state-independent manner, we can integrate over all unitaries $O$ supported on $A$, as follows:
\begin{definition}
 We define the  average state-dependence as 
\begin{equation}\label{DefdV}
        d_{\nu}(V) = \frac{1}{vol(\nu)}\int d\nu({O}) \Big[ \min_{\tilde{O}\in U(B)} D_V(\tilde{O},O)\Big]
    \end{equation}
    where $\nu(O)$ is the Haar measure. 
\end{definition}
\noindent By construction, $d_{\nu}(V)$ vanishes if and only if the reconstruction for all $O$ is exactly state-independent. From Theorem~\ref{thm:exactrecon}, we know that when the map is isometric, there is always a state-independent reconstruction for any unitary operator, hence this measure always vanishes for isometries.

Recall that for a map $V:A\rightarrow B$ we can always redefine a surjective version $V': A\rightarrow B'=Im(V)\subset B$ where $V,V'$ are related by a unitary transformation $U_B$ on $B$ such that $V=U_B V'$. By Proposition~\ref{prop:cDv}, we must have $\mathcal{D}(V')=\mathcal{D}(V)$. Furthermore, because any unitary reconstruction $\tilde{O}'$ on $B'$ can be identified with $\tilde{O}=U_B \tilde{O}' U_{B}^{\dagger}$ on $B$, the degree of state-dependence is also equal for $V$ and $V'$. Therefore, without loss of generality, it suffices to focus on the scenarios where $V:A\rightarrow B$ is surjective\footnote{Indeed for any non-surjective $V$, like ones where $|A|<|B|$, we can perform the above transformation and study the non-isometricity/state-dependence measure of $V'$ which would return the same value. } with $|A|\geq |B|$. The definition  of $d_\nu$ can be generalized to only average over specific operators or states as needed.
%Since for any case where $V$ is not onto, we can always identify a subspace $B'$ such that $V:A\rightarrow B'$ is surjective, it suffices to consider  
%\CC{ can sum over different subsets instead of just haar}
%In the same way, we have chosen to sum over the entire unitary group and state over $A$. 
%This definition  of $d_v$ can be generalized to only average over specific operators or states as needed.
%by the specific circumstances such as ones that are more natural to quantum gravity. We do not that in this section because we work with $V$ in complete generality with no additional information to bias one set of state or operator over another.

The minimization of $D_V(\tilde{O},O)$ over reconstructions $\tilde{O}$ that is required to measure state-dependence is difficult to carry out in general.  However, we can find bounds by using the singular values of $V$, as follows.
%The degree of state-dependence depends on the singular values of $V$. Although the minimization of $D_V(\tilde{O},O)$ is hard to do in general, one can bound the state-dependence measure as follows.
\begin{theorem}\label{thm:dvbound}
Suppose $V:A\rightarrow B$ has nonzero singular values $\{\sigma_i\}$ and  define
\begin{align}
f_U(\{\sigma_k\},|A|,|B|)&\equiv 2|A|\sum_{k=1}^{|B|}\sigma_k^2-\frac{2}{|B|}(\sum_{k=1}^{|B|}\sigma_k)^2\\
f_L(\{\sigma_k\},|A|,|B|)&\equiv \max\{0,2\Big[|A|(\sum_{k=1}^{|B|}\sigma_k^2)-\sqrt{|A||B|}(\sum_{k=1}^{|B|} \sigma_k)^2\Big]\} \, .
\end{align} 
Then the average state-dependence satisfies $f_L\leq  d_{\nu}(V)\leq f_U$. 
%\CC{removed a factor of $|A|$ as the note missed it.}
\end{theorem}
\begin{proof}
    See Appendix~\ref{app:proof_thm2.2}.
\end{proof}

\begin{definition}
A linear map $V$ is a co-isometry if $V^{\dagger}$ is an isometry.
\end{definition}
\noindent To preserve the inner product, all singular values of an isometry are equal to one. Therefore, from the singular value decomposition, the \emph{nonzero singular values} of its conjugate transpose, i.e., a co-isometry, are also equal to one. 
Note that sometimes we are interested not only in the co-isometries themselves, e.g., \cite{cao2023overlapping}, but also maps that are co-isometries rescaled by a constant, e.g., \cite{akers2022black}. A tighter bound on the average state-dependence can be obtained for such maps whose non-zero singular values are all equal to each other. 

\begin{theorem}\label{thm:coiso_bound}
    For $V$ proportional to a co-isometry such that the nonzero singular values are  $\sigma_k=\eta$ for all $k=1,\dots, |B|$, then $$2\eta^2(|A||B|-|B|^2)\geq d_{\nu}(V)\geq \max\{0,2\eta^2(|A||B|-|A|^{1/2}|B|^{2})\}.$$
%    where the upper bound is tight.
\end{theorem}
\begin{proof}
   See Appendix~\ref{app:proof_thm2.3}. 
\end{proof}
\noindent Here $\eta$ is the proportionality constant relating $V$ to a co-isometry.

The above measure quantifies the total state-dependence of all logical operators. Sometimes it may be sufficient or more practical to check whether a subset of operators can be reconstructed state-independently. This could for example be the set of simple operators (i.e., operators of polynomial complexity) or localized operators depending on the context. Here let us examine state-dependence by looking at a set of unitary operators that form an irreducible representation of the generalized Pauli group which forms a complete operator basis over $A$. This will allow us to explicitly compute the state-dependence measure for this set of operators, and make contact between our measures of non-isometry and state-dependence.
%{\color{red} What does it mean for an operator basis to be ``nice''?}
%\CC{Niceness is a detail we don't care about here, so I have removed the ref. it's more thoroughly explained in the original ref \cite{nicebasis}, but basically the basis is a irreducible projective unitary representation of a group.}   
%A unitary 1-design is an example of a nice basis, and allows us to sum over a discrete set of points instead of integrating over the entire unitary group manifold using the Haar measure.  

%\CC{My apologies. Despite their popularity in QI, I remembered that these matrices are not necessarily well-known in hepth where we usually want to think about generalized Gell-Mann matrices. Modified the text here to explain them better, and removed any specialized QI language.}
The generalized Pauli group is a group of matrices that are simply higher dimensional versions of Pauli matrices with the exception that they are not Hermitian in general\footnote{This is not to be confused with generalized Gell-Mann matrices which are Hermitian.}.  Consider $q\times q$ unitary matrices $X,Z$ (also known as shift and clock matrices), where the eigenstates of $Z$ are the computational basis states $\{|j\rangle\}$ such that $$Z|j\rangle = \omega^j|j\rangle$$ with $j\in \mathbb{Z}_q$ and $$\quad X|j\rangle = |j+1\rangle$$ where $+$ is modular addition.

\begin{definition}
    The generalized Pauli group (or Pauli group for short) over any Hilbert space of dimension $q<\infty$ is generated by the identity $I$ with global phase, and the $X,Z$ operators: $$\mathcal{P}=\langle \omega I, X, Z\rangle$$ where $\omega=\exp(2\pi i/q)$ for $q$ odd and $\exp(\pi i/q)$ for $q$ even. 
\end{definition}
\noindent Here $X$ and $Z$ are non-Hermitian unless $q=2$.  Note that the generalized Pauli group and unitary conjugations form a 1-design and so does any orthogonal unitary basis under the trace inner product \cite{Adam2013thesis,kaznatcheev2010structure}. %The Pauli group is generated by the identity and two unitary operators $X,Z$.  The eigenstates of $Z$ are the computational basis states $\{|j\rangle\}$ such that $$Z|j\rangle = \omega^j|j\rangle,\quad X|j\rangle = |j+1\rangle$$ where $+$ is modular addition. They are also known as the clock and shift matrices in some settings.    
%{\color{red} we are talking about the action of X,Z on computational basis states.  Does this mean that we are thinking about qubit systems only here?  If we want to think about more general systems, how are we supposed to use this qubit language?  Can we clarify?} \CC{no we are picking a basis over a hilbert space of any finite dimension $q$. Computational basis in the context of general dimensions does not exclusively refer to the binary qubit basis, but rather the basis $|j\rangle$ defined above. Clarified the index range of $j$ above.}

\begin{theorem}\label{thm:d1X0}
Let $V:A\rightarrow B$ be a linear map  and let $\chi_0$ be a unitary 1-design over $A$  such that $\chi_0=\{QPQ^{\dagger}: P\in \mathcal{P}, \, Q^{\dagger}|\psi_i\rangle=|i\rangle\}$ where $\mathcal{P}$ is the generalized Pauli group over $A$, $\{|\psi_i\rangle\}$ is the eigenbasis of $V^{\dagger}V$, and $\{|i\rangle\}$ is the computational basis spanning $A$. Then
 \begin{equation}\label{Defd1}
        d_1(V,\chi_0) \equiv \frac{1}{|\chi_0|}\sum_{O\in \chi_0} \Big[ \min_{\tilde{O}\in U(B)} D_V(\tilde{O},O)\Big] = \frac{1}{|A|}\sum_{i,j=1}^{|A|}(\sigma_i-\sigma_j)^2
    \end{equation}
    where $\{\sigma_i\}$ are the singular values of $V$ and $|\chi_0|$ is the number of elements in the set $\chi_0$.  For the basis vectors that span the null space we take the singular value to be $0$.
\end{theorem}
\begin{proof}
    See Appendix \ref{sec:thmd1X0}.
\end{proof}

Note that not all operators are reconstructable even in a state-specific way. Non-reconstructable operators include, e.g., operators that rotates a state in the kernel to a state outside the kernel. We expand on this point after proving the above theorem in Appendix~\ref{sec:thmd1X0}. While we deal with global non-isometries and related state-dependent reconstructions~\cite{Akers:2020pmf,Akers:2021fut}, similar observations are made for subregion reconstructions. In particular, it is noted that not all operators are reconstructable even state-specifically \cite{Akers:2023fqr}. 

When $V$ is proportional to a co-isometry with non-zero singular values equal to $\eta$, we see that $d_1(V,\chi_0)=2\eta^2(|A||B|-|B|^2)$ which saturates the upper bound of $d_{\rm \nu}(V)$ in Theorem~\ref{thm:dvbound}.  For any 1-design generated by left unitary multiplication, $\chi(U)=\{US, S\in \chi_0\}$,  $d_1(V,\chi_0)\geq d_1(V,\chi(U))$.
We prove this in Appendix~\ref{app:proof_thm2.2} and Appendix~\ref{app:proof_thm2.3}. In other words, $d_1(V,\chi_0)\geq d_{\nu}(V)$ constitutes an upper bound for the Haar average state-dependence if $V$ is proportional to a co-isometry.

Note that since we are not averaging over the entire unitary group when we sum over operators $O$ in the measure $d_1$, there are many choices of 1-designs one can consider. The above shows that other than the technical convenience it provided for the proof, the particular 1-design $\chi_0$ we chose maximizes the total state-dependence. Intuitively, it picks out a set of elements that are most sensitive to the state-dependent nature of $V$ for operator reconstruction. 
%{\color{red} But the one-design is only an approximation of the state-dependence. Is this saying that it maximizes the state-dependence among one designs? Why is that interesting?  Wouldn't I want to minimize the state-dependence to get a better upper bound?}\CC{added explanation. I would think the opposite is true--- we would like to pick a set of operators that are most sensitive to the state-dependence so we don't get fooled into thinking that a state-dependent $V$ has a small or zero measure $d_1$. In any case, that's such a minor point in our narrative, it's just a math thing for me to characterize some properties of $\chi_0$. If it's distracting I can move it to the appendix.}
%\CC{there are some typos and constants that might need to be fixed. Go thru afterwards.}

%Just like the non-isometry measure, this state-dependence measure can be similarly refined by writing $d_1(V,\chi_0)=d_1^{(W)}(V,\chi_0)+d_1^{(S)}(V,\chi_0)$ such that
Our state-dependence measure can be  refined to separate pieces that reflect a strong and weak state-dependence, as follows. 
\begin{definition}[Strong and weak state-dependence]
Let $Z^c$ be a set that indexes the non-zero singular values of $V$ and $Z$ a set that indexes the set of zero singular values. (States in the kernel of $V$ are assigned a zero singular value.)
Define
\begin{align}
    d_1(V,\chi_0)&=d_1^{(W)}(V,\chi_0)+d_1^{(S)}(V,\chi_0) \\
    d_1^{(W)}(V,\chi_0)&:=\frac{1}{|A|}\sum_{i,j\in Z^c}(\sigma_i-\sigma_j)^2\\
    d_1^{(S)}(V,\chi_0)&:=\frac{2}{|A|}\sum_{i\in Z,j\in Z^c}(\sigma_i-\sigma_j)^2
    %:=\frac{2}{|A|}\sum_{i\in Z, j\in Z^c} (\sigma_i-\sigma_j)^2.
    \label{eqn:swdependence}
\end{align}
We say  $d_1^{(W)}$ quantifies weak state-dependence while $d_1^{(S)}$ quantifies strong state-dependence.
\end{definition}
\noindent We have separated the two contributions to the state-dependence measure  in anticipation of results in the next section. Qualitatively, it is clear that $d^{(S)}_1$ receives contributions only if the map $V$ has a non-trivial kernel, because it depends on having vanishing singular values.  This in turn implies strong non-isometry. Meanwhile,  $d_1^{(W)}$ is non-vanishing if the nonzero part of the singular value spectrum is non-flat.  This in turn says that some basis elements are stretched relative to others by $V$, implying weak non-isometry.
% We will refer to thse contributions as coming from strong and weak state-dependence respectively.

Let's simplify the expression in (\ref{eqn:swdependence}) by summing over $i\in Z$. If we further impose the normalization condition $\Tr[V^{\dagger}V]=|A|$ similar to \cite{akers2022black}, then we see that
\begin{equation}
        d_1^{(S)}(V,\chi_0) = \frac{2|Z|}{|A|} \Tr[V^{\dagger}V]=2\dim(ker(V)).
\end{equation}
where strong state-dependence only depends on the size of the kernel.

\subsection{Connecting state-dependence with non-isometry}
%Since we have established $d_1(V,\chi_0)$ as a reasonable and computable estimate of $d_{\nu}(V)$. let us compare it with the non-isometry measure $\mathcal{D}(V)$ (Def.~\ref{def:nonisometry_measure}). First we notice that both are quantities that are measuring some form of non-flatness of the singular value spectrum of $V$. In this physical sense, it is no surprise that they are probing more or less the same thing. 

We can now compare  $d_1(V,\chi_0)$, which is a computable estimate of $d_{\nu}(V)$, with the non-isometry measure $\mathcal{D}(V)$ (Definition~\ref{def:nonisometry_measure}).   Qualitatively,  both  quantities  are measuring some form of non-flatness of the singular value spectrum of $V$, and so they are clearly related. Below we will make this connection more precise, first in a simple case, and then generically.

\begin{proposition}
\label{prop:noniso=statedep}
Let $V:A\rightarrow B$ be proportional to a co-isometry while satisfying the normalization constraint $\Tr[V^{\dagger}V]=|A|$,
then 
\begin{equation}\frac{\dim(ker(V))}{|A||B|}=\frac{|A|-|B|}{|A||B|}=\frac{1}{2|A||B|}d_1(V,\chi_0)=\mathcal{D}(V).\end{equation}  
\end{proposition}
This follows from Lemma~\ref{lemma:non_isometry} and Theorem~\ref{thm:d1X0} by remembering that all non-zero singular values of a co-isometry are equal to each other. In other words, for co-isometric maps,  non-isometricity is quantitatively equivalent to state-dependence. 
%{\color{red} The sentence here said  ''We see that they are actually equivalent for a class of maps $V$ when $|A||B|$ is a constant.''  I edited because I was not sure what the qualification  ``when $|A||B|$ is a constant'' meant.  Of course, it's constant, since it's just a number.}\CC{yeah im fine getting rid of that. INdeed it's always a number, but it's a different number for different classes of maps. Any two numbers are proportional to each other by constant multiplication, so just saying that they are prop to each other with some number is not good enough. This is saying though that for a class of maps, these measures are proportional to each other by the same const multiple. In any case, let's get rid of it since it has confused more than one person so far.}

%This is basically the story we were trying to tell --- non-isometricity gives rise to state-dependence and vice versa. We see that they are actually equivalent for a class of maps $V$ when $|A||B|$ is a constant. 

%A similar albeit more general relationship holds. For example, by slightly rewriting $\mathcal{D}(V)$ in (\ref{eqn:non-iso-square}), it is easy to read off that

In fact, we can state a more general relationship by  rewriting $\mathcal{D}(V)$ in (\ref{eqn:non-iso-square}) as
\begin{equation}\label{eqn:dDbound}\frac{\lambda_{\rm min}}{2|A|^2} d_1(V,\chi_0)\leq \mathcal{D}(V)\leq  \frac{4\lambda_{\rm max}}{2|A|^2} d_1(V,\chi_0)\end{equation}
where $\lambda_{\rm min},\lambda_{\rm max}$ are respectively the minimum and maximum nonzero eigenvalues of $V^{\dagger}V$.

One can arrive at this bound by recognizing that $\sigma_i^2=\lambda_i$ and expanding the expression in Lemma \ref{lemma:non_isometry} using a difference of squares: $(\sigma_i^2-\sigma_j^2)^2 = (\sigma_i+\sigma_j)^2(\sigma_i-\sigma_j)^2$. Let $\sigma_{\rm max}, \sigma_{\rm min}$ be the maximum and minimum nonzero singular values of $V$. % after normalization, i.e. $\sum_{i}\lambda_i^2=|A|$. Then it is clear that $\lambda_{\rm min}>0$ by definition and $\lambda_{\rm max}^2<|A|$. %Note that this normalization is what Akers and Pennington used for the random projection map, which tries to rescale the magnitude of vectors under the map such that the overall inner product is preserved on average. 
Because for any  $\sigma_i,\sigma_j$ such that at least one of them is not zero, $\sigma_i+\sigma_j\leq 2\sigma_{\rm max}$ and $\sigma_i+\sigma_j\geq \sigma_{\rm min}$, we must have

\begin{equation}
    \sigma_{\rm min}^2\sum_{i,j}(\sigma_i-\sigma_j)^2 \leq \sum_{i,j}(\sigma_i^2-\sigma_j^2)^2\leq 4\sigma_{\rm max}^2\sum_{i,j}(\sigma_i-\sigma_j)^2.
\end{equation}
Note that for the terms where both eigenvalues are 0, $\sigma_{\rm min}(\sigma_i+\sigma_j) =0$ anyway and they don't contribute to the sum. Restoring the coefficients from Lemma~\ref{lemma:non_isometry} and Theorem~\ref{thm:d1X0}, we obtain the inequality above.

%\noindent{\color{red} Can we add a sentence explaining how the rewriting of $ \mathcal{D}(V)$ allowed us to show this? }\CC{added a proof}
In other words, \emph{any amount of non-isometry will  lead to some state-dependence and vice versa}.  We can derive an even tighter bound by splitting the strongly and weakly non-isometric/state-dependent components. It is easy to see that by grouping the sums in $\mathcal{D}(V), d_1(V,\chi_0)$ over $i,j$ into sums where $i,j\in Z$ and sums where $i\in Z, j\in Z^c$ or $i\in Z^c, j\in Z$, one obtains 
\begin{equation}
\frac{2\lambda_{\rm min}}{|A|^2}d_1^{(W)}(V,\chi_0)\leq\mathcal{D}_{W}(V)\leq \frac{2\lambda_{\rm max}}{|A|^2}d_1^{(W)}(V,\chi_0)\end{equation}
and 
\begin{equation}\frac{\lambda_{\rm min}}{2|A|^2}d_1^{(S)}(V,\chi_0)\leq \mathcal{D}_{S}(V)\leq \frac{\lambda_{\rm max}}{2|A|^2} d_1^{(S)}(V,\chi_0).\end{equation}

%Heuristically, we can look at the above as two types of contributions to non-isometry: the unequal singular values which contribute to the stretching and shrinking of vectors which can happen without a kernel (in weak non-isometries) and the size of the kernel
%(in strong non-isometries).
%\CC{copied from note. smooth later}
 \noindent\emph{Therefore, strong (weak) non-isometry implies strong (weak) state-dependence and vice versa.} 

 Recombining these equations we find that the total non-isometry is lower and upper bounded by state-dependence measures:
\begin{theorem}
Let $V:A\rightarrow B$ be a linear map, then
\begin{equation}
   \frac{2\lambda_{\rm min}}{|A|^2}(d_1^{(W)}(V,\chi_0)+\frac 1 4 d_1^{(S)}(V,\chi_0))\leq \mathcal{D}(V)\leq \frac{2\lambda_{\rm max}}{|A|^2}(d_1^{(W)}(V,\chi_0)+\frac 1 4 d_1^{(S)}(V,\chi_0))
\end{equation}
\end{theorem}
\noindent These bounds are tighter when $\lambda_{\rm min}\approx \lambda_{\rm max}$, i.e., when $V$ is  proportional to an approximate (co-)isometry. 
%{\color{red}   Fo co-siometries I thought we said that all the eigenvalues are actually equal to 1. So in that case, aren’t we done?   Or is the “approximately propotional” nature the reason why the eigenvalues are not exactly 1 here? }
%\CC{the second one, we are talking about the case where they are close, not but exactly 1. CHanged approximately prop to to prop to approx (co)-isometry. Like before, proportional to means the eigenvalues are not 1, but is a constant times 1. Indeed your comment is in the next sentence. } 
In fact, when $\lambda_i=\lambda_j=\lambda,~\forall \lambda_i,\lambda_j\ne 0$  we recover Proposition~\ref{prop:noniso=statedep} as a corollary.  For approximate holographic codes introduced by \cite{Hayden_2016,Cao_2021} as well as the non-isometric codes of \cite{akers2022black,cao2023overlapping},  the above bounds are either saturated or close to being saturated. 
%\CC{comment on instances where you have near-flat spectrum, such as near vacc holographic code, non-isometric code for BH, for approxi QECC in practice, these bounds are very close to being equalities.}

%\CC{got rid of the remark structure}

For maps that are co-isometries up to constant rescaling, since $d_{\nu}(V)\leq d_1(V,\chi_0)$, we have 
\begin{equation}
    \frac{\lambda_{\rm min}}{2|A|^2} d_{\nu}(V)\leq \mathcal{D}(V).
\end{equation}
As we don't have an upper bound of $d_1(V,\chi_0)$ in terms of the Haar averaged measure, we leave the upper bound on $\mathcal{D}(V)$ and the generalization beyond co-isometric maps for future work.

Note that the results above did not require any statements about quantum error correction. In  holographic quantum error correcting codes and, indeed, in every finite dimensional approximate quantum error correcting code the authors can think of, the spectrum is flat or nearly flat \cite{gottesmannote,ErrorCorrectionZoo,Hayden_2016,Cao_2021,Akers:2021fut,akers2022black},
%{\color{red} citations?}\CC{I thought about adding something, but then i dont want to cite all works in QEC. Maybe I cite the error correction zoo or Gottesman lecture note, is that ok with everyone? QECCs are defined to be (partial) isometries, and the only mentions of non-isometric ones (based on post-selection) are from ads cft, which I will cite here.} 
meaning that the difference between $\lambda_{\rm min}$ and $\lambda_{\rm max}$ is very small or zero. Therefore, for all these systems we expect our bounds to collapse to equalities precisely relating non-isometry to state-dependence.
 
%{\color{red} But we started this section saying that the non-isometry measure and the state-dependence measure are both measuring some form of non-flatness. But then we are saying that the spectrum is flat in the cases of interest.  So doesn't that mean that in addition to collapsing the bounds to equality, the state-dependence and non-isometry will actually vanish?}
%\CC{no, for weak non-isometry this is true, but when the kernel is non-trivial, one can have a flat spectrum for the non-zero singular values but still retain a non-trivial state-dependence measure prop to the size of the kernel. We have this in one of the corollaries above I believe, but we can add a bit more explanation here.}

%It's interesting to note that to this point, the error correction aspect of the story has not been invoked, only the state-dependent nature of the map. In the holographic quantum error correcting code (and, indeed, in every approximate quantum error correcting code the authors can think of), the spectrum is flat or nearly flat, meaning that the difference between $\lambda_{\rm min}$ and $\lambda_{\rm max}$ is very small or zero. Therefore, we believe that for the systems in question, that these bounds collapse to equalities precisely relating non-isometry to state-dependence in a quantum error correction context.

So far we have been using a  definition where state dependence means that no single ``physical'' operator $\tilde{O}$ can exactly perform the same ``logical'' operation as $O$  for all states $|\psi\rangle$. However, we could alternatively define state dependence as a situation where a ``logical operator'' $O$ does not admit the same best physical reconstruction for all states. In other words,  the optimal reconstruction $\tilde{O}$ of $O$ may be different for different states. If we allow the reconstruction to be imperfect, the quantification of the amount of state dependence that we defined earlier in this section does not characterize the amount of the alternative form of state dependence we just defined. For example, suppose the best possible operator reconstruction $\tilde{O}$ one can achieve state-specifically is the same for all states, but the size of $||\tilde{O}V|\psi\rangle-VO|\psi\rangle||$ differs from state to state. Then, this operator is clearly not state-independently reconstructable according to Definition~\ref{def:stateindependence}, but  the reconstruction $\tilde{O}$ does not look like it depends on the states at all. For example, this is the case for the operators considered in the proof of Theorem \ref{thm:d1X0} in Appendix \ref{sec:thmd1X0}, at least for the basis states in which $V$ is diagonal. Instead, we can define an alternative state dependence measure:
\begin{definition}
The operator variance of $O$ with respect to a linear map $V$ is 
    \begin{equation}
       \mathcal{V}_V(O)=\min_{\tilde{O}\in U(B)}D_V(\tilde{O},O)-\int d\mu(|\psi\rangle)\min_{\tilde{Q}\in U(B)} ||\tilde{Q}V|\psi\rangle - VO|\psi\rangle||^2
       \label{eqn:statevariance}
    \end{equation}
    where $\mu(|\psi\rangle)$ is the Haar measure.
\end{definition}
\noindent In this definition, we subtract off any non-varying contribution like in the example above such that the leftover contribution corresponds to the amount of true ``dependence''. Since the minimization of the second term in equation (\ref{eqn:statevariance}) is done inside the integral, it is clear that $\min_{\tilde{O}\in U(B)}D_V(\tilde{O},O)\geq \int d\mu(|\psi\rangle)\min_{\tilde{Q}\in U(B)} ||\tilde{Q}V|\psi\rangle - VO|\psi\rangle||^2$, i.e. $\mathcal{V}_V(O)\geq 0$. A sum similar to Definition \ref{DefdV} or Theorem~\ref{thm:d1X0} can also be performed over $O$ if operator averaging is needed. We will leave the analysis of this alternative measure and its relation with our non-isometry measure for future work.

%It is interesting to note that $d_1(V,\chi_0)$ can be rewritten as a sum of two parts, one that computes the dimension of the kernel, and a contribution that comes from the non-flatness of the spectrum of $V$. This is roughtly consistent with our physical intuition that non-isometry or state-dependence can arise from both the stretching of vectors even when $|A|\leq |B|$ and from the non-trivial kernel when $|A|>|B|$. Generally speaking, this distinction is somewhat artificial as one can think of the kernel as shrinking the length of vectors to zero.

%\CC{Do we want to say more about the implication of the above results?}

\section{Non-isometry and state-dependence in holography}
\label{sec:gravity}

Our goal in this section is to interpret some of the general results arrived at above within the framework of holographic duality.  We will also explore the interplay between state-dependent bulk reconstruction, non-isometric bulk-to-boundary maps, and global causal structure of the bulk spacetime. Note that here, and in the rest of the paper, we use the term ``bulk-to-boundary map'' as shorthand for the map between the bulk effective field theory (EFT) description and the fundamental, UV-complete description in the dual theory.\footnote{Of course the full holographic duality maps {\it all} degrees of freedom of the bulk into the dual boundary theory, but we will be interested in the restriction of the domain of this map to a bulk EFT.} Some recent work has considered scenarios in  which auxiliary, non-gravitating degrees of freedom are either coupled to, or adjoined to, the holographic theory \cite{almheiri2019entropy,Penington:2019kki,almheiri2020replica}. In these cases, by ``bulk'' we will mean  the EFT description of the spacetime and matter plus the low-energy EFT description of the auxiliary variables, while by ``boundary'' we will mean the microscopic description of both the holographic theory and the auxiliary DOF.
%{\color{blue}In situations where auxiliary,  non-gravitational degrees of freedom (DOF) are present in addition to a holographic theory (either coupled \cite{}\sa{Add citations} or not \cite{}), by ``bulk'' we will mean by the spacetime, including bulk matter, dual to the holographic theory \textit{plus} the low-energy EFT description of the auxiliary DOF.   Meanwhile, by ``boundary'' we will mean the microscopic description of both the holographic theory and the auxiliary DOF.} 

We focus on global reconstruction and global bulk-to-boundary maps, as in Sec.~\ref{sec:QI}. By ``global'' we mean that we are including in our description all connected and disconnected components of both the bulk and the boundary theories.  This implies that all states both in the bulk theory and in the dual microscopic theory can be taken to be pure.  This is because, if we want to consider a mixed state, we can simply purify it by entangling it with some auxiliary degrees of freedom, and then consider the whole system.   While we focus on AdS/CFT, we expect our arguments to extend to recently proposed scenarios where a compact region of spacetime may encode, through underlying entanglement, the local physics of a larger region \cite{Bousso:2022hlz,Bousso:2023sya,Balasubramanian:2023dpj}.

%For simplicity, we will focus our attention on AdS/CFT settings, although we do not see significant obstacles preventing the extension of our arguments to more generic holographic setups.

\subsection{Strong vs weak non-isometry in holography}
\label{sec:weakvsstrong}

Multiple lines of evidence suggest that the holographic bulk-to-boundary map can manifest both strong and weak non-isometry. To see this, let us consider the bulk-to-boundary map from the bulk effective field theory (EFT) Hilbert space, namely a low-energy approximation of the full quantum gravity defined with respect to some cutoff, and the fundamental boundary Hilbert space.\footnote{Note that when defining the bulk EFT Hilbert space, we typically ignore non-perturbative quantum gravity effects such as contributions to the gravitational path integral from higher topology saddles. As we will see, this ``naive'' definition of the bulk EFT is exactly what leads to non-isometry, because these saddles give non-trivial contributions to the inner product in the fundamental theory.} Some states in this Hilbert space are well-described as perturbative field excitations around a semiclassical background. We typically treat such states as basis elements for the EFT. Note that not all states in the resulting Hilbert space have semiclassical descriptions---e.g., a superposition of two very different geometries \cite{Balasubramanian:2007zt}, or a superposition of exponentially many perturbative excitations \cite{akers2022black} are not expected to have a good semiclassical description.\footnote{In the AdS/CFT literature, the EFT Hilbert space is sometimes called the ``code subspace'' because, when the bulk-to-boundary map is isometric, it is isomorphic to its image in the (physical) boundary Hilbert space, which is the properly defined code subspace. However,  this nomenclature is misleading if the map is strongly non-isometric, because the map has a kernel, and the EFT states are overcomplete.   To avert confusion we avoid mention of code subspaces, and instead
specify if we are talking about the bulk EFT Hilbert space or a subspace of the boundary Hilbert space, the latter understood as the physical non-redundant description of the physics. 
} We will denote the subset of states in the bulk EFT Hilbert space which admit a semiclassical description by $S_{SC}$. Note that $S_{SC}$ does not form a subspace in general as superpositions of states in this set need not remain semi-classical.

Now consider a black hole with horizon area $\mathcal{A}_h$ and Bekenstein-Hawking entropy $S_{BH}=\mathcal{A}_h/4G$.  
%We could also consider a two-sided, eternal black hole, for which there are two horizons with this area, doubling the entropy.  
If  $S_{BH}$ has a standard interpretation as a coarse-grained entropy, we expect the underlying fundamental Hilbert space to have dimension $e^{S_{BH}}$.  However, it is easy to see that in the EFT the apparent dimension can be much larger. That is because black holes have large interiors, and the Hilbert space of local EFT degrees of freedom subject to a cutoff can be very big, even if we restrict just to states with good semiclassical descriptions.  Recent work has shown that this tension is resolved by nonperturbative quantum gravity effects that induce small overlaps between states that are orthogonal in the EFT, implying non-isometry \cite{penington2022replica,almheiri2020replica,akers2022black,Balasubramanian:2022gmo, Balasubramanian:2022lnw,Climent:2024trz, Balasubramanian:2024yxk}.

%that are orthogonal in the EFT, such as different geometries \cite{Balasubramanian:2022gmo, Balasubramanian:2022lnw,Climent:2024trz, Balasubramanian:2024yxk} or perturbative  excitations behind the horizon \cite{Antonini:2023hdh}, and  states in toy models with constructs like end-of-the-world branes of different flavors \cite{Penington:2019kki}. {\color{red}  (more cites?)}  Thus, the bulk-to-boundary map  takes  orthogonal states in the EFT to overlapping states in the fundamental, dual CFT description.  So it is at least weakly isometric.  But if we pick the cutoff so that the dimension of the EFT is greater than the exponential of the Bekenstein-Hawking entropy, the map must also have a kernel, and so is strongly isometric. 

For example, consider the explicit construction of semiclassically well-controlled AdS black hole microstates in \cite{Balasubramanian:2022gmo, Balasubramanian:2022lnw} as shells of matter of different masses propagating behind the horizon, and modifying the geometry there while the exterior remains fixed, see Figure \ref{fig:shells}. Consider a basis of $\Omega$ shell states $|\psi_p\rangle$ that are mutually orthogonal in the EFT
\begin{equation}
\bra{\psi_p}\ket{\psi_{p'}}_{\rm EFT}\equiv \bra{\psi_p}\ket{\psi_{p'}} =\delta_{pp'} \, ,
\end{equation}
because the shells have different inertial masses and positions behind the horizon. The authors of \cite{Balasubramanian:2022gmo, Balasubramanian:2022lnw} argued that, non-perturbatively, the overlap between two different shell states is 
\begin{equation}
    \bra{\psi_p}\ket{\psi_{p'}}_{\rm full} \equiv \bra{\Psi_p}\ket{\Psi_{p'}}= \delta_{pp'} + R_{pp'}  \, ,
    \label{eq:bulkEFTinner}
\end{equation}
where $R_{pp'}$ is a pseudorandom variable, as inferred from wormhole effects in quantum gravity. The second term arises from pseudorandom statistics of the heavy operators creating the shell states \cite{Belin:2020hea,Balasubramanian:2022gmo,Balasubramanian:2022lnw,sasieta2023wormholes,antonini2024holographic}. These statistics, which are an avatar of the Eigenstate Thermalization Hypothesis \cite{Belin:2020hea}, are coarse-grained over in saddle point computations of the gravitational path integral \cite{Belin:2020hea,Balasubramanian:2022gmo,Balasubramanian:2022lnw,sasieta2023wormholes,Antonini:2023hdh,antonini2024holographic, Chandra:2022bqq, Chandra:2023dgq}.

This means that orthogonal EFT states actually overlap in the fundamental theory.  In terms of the bulk-to-boundary map $V$ we can write this as 
\begin{equation}
\bra{\psi_p} V^\dagger V\ket{\psi_{p'}} = \delta_{pp'} + R_{pp'},  
\label{eq:mappedinner}
\end{equation}
where $V\ket{\psi_p}=\ket{\Psi_p}$.
So orthogonal EFT states can overlap in the dual CFT, implying, at a minimum, weak non-isometry of the bulk-to-boundary map.  The authors of \cite{Balasubramanian:2022gmo, Balasubramanian:2022lnw} further showed that if $\Omega < e^{S_{BH}}$, the CFT Gram matrix of overlaps has full rank so that the span has dimension $\Omega$.  But if $\Omega > e^{S_{BH}}$, namely the shell state basis size is chosen larger than the dimension of the underlying black hole Hilbert space, then linear combinations of basis elements include null states of zero norm (i.e., vanishing eigenvalues of the Gram matrix) in just the right number to ensure that the span has dimension $e^{S_{BH}}$. This happens because small overlaps between states, in this case neglected in the bulk EFT, can drastically reduce the dimension of the  Hilbert space \cite{cao2023overlapping, akers2022black, Akers:2021fut, balasubramanian2023quantum}. Thus, the bulk-to-boundary map, which maps orthogonal states in the naive EFT to non-orthogonal states in the fundamental theory, must have a kernel, and thus be strongly non-isometric.  Just as in \cite{akers2022black}, the non-isometricity of the map between the bulk EFT and the fundamental description arises when we consider a naive bulk EFT in which non-perturbative corrections to the inner product have been neglected.

A completely analogous conclusion can be reached by constructing an overcomplete basis of black hole microstates just in terms of EFT matter excitations around a fixed background geometry \cite{penington2020entanglement, almheiri2019entropy,Penington:2019kki,almheiri2020replica,akers2022black}. These will once again acquire non-perturbative corrections to the overlaps leading to a non-isometric bulk-boundary map. The shell state basis of \cite{Balasubramanian:2022gmo, Balasubramanian:2022lnw} and such an EFT excitation basis can presumably then be written as linear combinations of each other, but to determine the required superposition phases we would need control of the complete non-perturbative quantum gravity.  Recent work investigating the Page curve for black hole evaporation \cite{penington2020entanglement, almheiri2019entropy,Penington:2019kki,almheiri2020replica,akers2022black} effectively uses such perturbative states, with the map switching between weak and strong non-isometry after the Page time. Similar states also appear in studies of Python's lunch geometries in AdS/CFT \cite{brown2020python,Engelhardt:2021mue, Engelhardt:2021qjs}. 
In this case the bulk EFT Hilbert space is typically taken to be  smaller than the fundamental Hilbert space, with the bulk-to-boundary map being therefore weakly-non-isometric.

Comparable results can also be obtained for closed universes entangled with AdS spacetimes \cite{Antonini:2023hdh}, in which the overlaps of EFT states receive corrections suppressed by the entanglement entropy between bulk fields in the closed universe and bulk fields in the AdS spacetimes. These are $O(G^0)$, but nevertheless non-perturbative in the EFT, and can produce weak or strong non-isometry depending on how large the EFT Hilbert space is taken to be. The same phenomenon occurs in toy models of black hole microstates involving End-Of-The-World branes behind the horizon \cite{Kourkoulou:2017zaj,Almheiri:2018ijj,Penington:2019kki,Cooper:2018cmb,Antonini:2019qkt,Antonini:2021xar,Waddell:2022fbn,Antonini:2024bbm,Geng:2024jmm}.

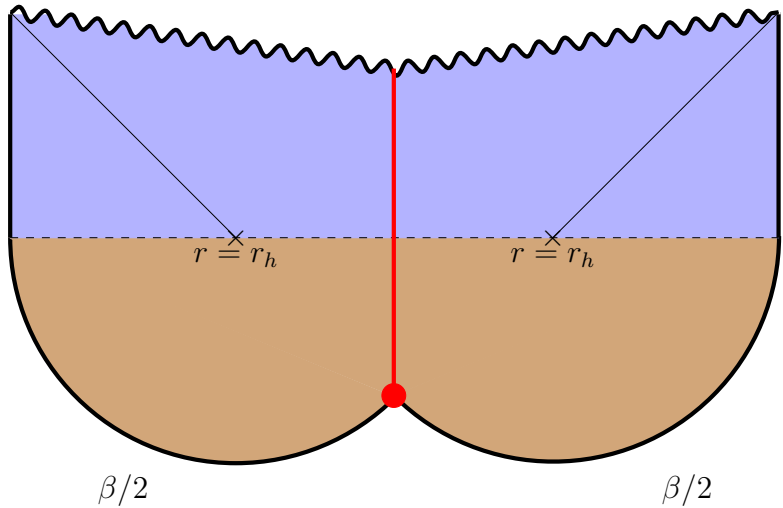
\begin{figure}[ht]
    \centering\begin{tikzpicture}[scale=1.5]
       \filldraw[color=brown!70] (0,0)--(6.81,0)--(3.4,-1.4)--cycle;
       \draw[ultra thick,fill=brown!70] (0,0) arc[start angle=-180, end angle=-45, radius=2];
       \draw[ultra thick,fill=brown!70] (3.4,-1.4) arc[start angle=-135, end angle=0, radius=2];
       \draw[thick, color=brown!70] (3.4,-1.4)--(0,0);
       \filldraw [red] (3.4,-1.4) circle (3pt);
       \filldraw[color=blue!30](6.81,2) rectangle (0,0);
       \draw[ultra thick] (0,0)--(0,2);
       \draw[ultra thick] (6.81,0)--(6.81,2);
       \draw[ultra thick,color=white] (6.81,2) rectangle (0,2.2);
       \draw[dashed, ultra thin, fill=brown!70] (0,0)--(6.81,0);
       \draw[ultra thin] (0,2) -- (2,0);
       \draw[ultra thin] (6.81,2) -- (4.81,0);
       \node at (4.81,0) {$\times$};
       \node[anchor=north] at (4.81,0) {$r=r_{h}$};
       \node[anchor=north] at (1,-2) {$\beta/2$};
       \node[anchor=north] at (6,-2) {$\beta/2$};
       \node at (2,0) {$\times$};
       \node[anchor=north] at (2,0) {$r=r_{h}$};
       \draw[ultra thick, decorate, decoration=snake,fill=white] (0,2)-- (3.4,1.5) -- (6.81,2);
       \draw[ultra thick, red] (3.4,-1.4)--(3.4,1.5);
        \end{tikzpicture}
    \caption{Geometric black hole microstate with a spherically-symmetric shell of pressureless matter behind the horizon. In the bottom half of the figure we show the Euclidean preparation of the microstate via the gravitational Euclidean path integral. In the top half we depict the corresponding Lorentzian geometry. We are considering here the case where the two exterior geometries are identical, but generically the two black holes can be at different temperatures \cite{Chandra:2022fwi, Balasubramanian:2022gmo, Climent:2024trz}. Shells of different mass, evolved for accordingly different amounts of Euclidean times, give rise to identical exterior geometries but different interior geometries, and therefore give rise to different black hole microstates.}
    \label{fig:shells}
\end{figure}

\subsection{Approximate isometry and state-independent reconstruction}
\label{sec:approximate}

%\sa{One thing missing from this subsection is a relationship between the two definitions given here and the measures of state-dependence/non-isometry given in the QI section. I think that the ones here are more stringent (they are worst-case definitions rather than average ones). I think from a holography point of view, the ones I included here are more meaningful, but it would be nice to make a further connection between the QI and gravity section by employing the QI ones as well. Any ideas?}

As we have seen in Sec. \ref{sec:QI}, reconstruction is always state-dependent for both weakly and strongly non-isometric maps.\footnote{Note that in this section we will make use of stronger, worst-case measures for state-dependence and non-isometry with respect to the average ones used in Sec.~\ref{sec:QI} and in \cite{akers2022black}.} Moreover, if $V:A\to B$ is non-isometric, there are unitary operators acting on $A$ which cannot be exactly reconstructed by corresponding unitaries on $B$, not even in a state-specific manner. We already encountered such operators in the proof of Theorem \ref{thm:exactrecon}. Let us analyze them in more detail. Consider a non-isometric map $V$ and two states $\ket{\psi_1}$, $\ket{\psi_2}$ such that $V\ket{\psi_1}=l_1\ket{\Psi_1}$, $V\ket{\psi_2}=l_2\ket{\Psi_2}$, where $|\Psi_{1,2}\rangle$ and $|\psi_{1,2}\rangle$ are unit normalized with $l_1\neq l_2$. Now consider a unitary operator $O$ acting on $A$ such that $O\ket{\psi_1}=\ket{\psi_2}$. We would like to find a unitary operator $\tilde{O}$ on $B$ such that $VO\ket{\psi_1}=\tilde{O}V\ket{\psi_1}$. However, the norm squared of the left hand side is given by 
\begin{equation}
    \langle\psi_1|O^\dagger V^\dagger V O|\psi_1\rangle=l_2^2
\end{equation}
and the norm squared of the right hand side is
\begin{equation}
    \langle\psi_1|V^\dagger \tilde{O}^\dagger \tilde{O}V |\psi_1\rangle=l_1^2,
\end{equation}
where we used the unitarity of $\tilde{O}$. Since $l_1\neq l_2$, the unitary $\tilde{O}$ does not exist.

On the other hand, if we allow the reconstruction $\tilde{O}$ of the operator $O$ to be approximate such that
\begin{equation}
    ||(\tilde{O}V-VO)\ket{\psi_1}||<\epsilon
    \label{eq:recbound}
\end{equation}
for some $\epsilon\ll 1$, then a unitary $\tilde{O}$ reconstructing $O$ can exist for sufficiently small $|l_1-l_2|$. In fact, we have 
\begin{equation}
    ||(\tilde{O}V-VO)\ket{\psi_1}||=\sqrt{l_1^2+l_2^2-l_1l_2\left(\langle\Psi_2|\tilde{O}|\Psi_1\rangle+\langle\Psi_1|\tilde{O}^\dagger|\Psi_2\rangle\right)}\geq |l_1-l_2|,
    \label{eq:differentscalings}
\end{equation}
where the optimal choice of $\tilde{O}$ for which the equality holds is  $\tilde{O}\ket{\Psi_1}=\ket{\Psi_2}$. With this choice and further choosing $1\gg \epsilon>|l_1-l_2|$, equation \eqref{eq:recbound} is satisfied and $\tilde{O}$ is an approximate reconstruction of $O$ within the allowed error.

Similar reasoning can be applied to state-dependent reconstruction, leading us to define a notion of approximate reconstruction \cite{Akers:2021fut} as follows:
\begin{definition}
    Consider a map  $V:A\to B$ between Hilbert spaces, a set of states $S=\{\ket{\psi_i}\}\subseteq A$, and a unitary operator $O$ which only admits a state-dependent reconstruction when acting on $A$. The operator $O$ is said to be approximately state-independently reconstructible on $S$ if there exists a unitary operator $\tilde{O}$ acting on $B$ such that
    $$||(\tilde{O}V-VO)\ket{\psi_i}||<\epsilon \quad \quad \forall \ket{\psi_i}\in S,$$
    where $\epsilon\ll 1$ is a fixed parameter.
    \label{def:apprec}
\end{definition}

In the context of semiclassical holography,  the reconstruction of bulk EFT operators is  approximate, and obtained only up to non-perturbative quantum gravitational corrections of $O\left(e^{-\gamma/G}\right)$ where $G$ is Newton's constant and $\gamma>0$ is $O(1)$ \cite{Hayden:2018khn}. In this context there is a natural value for the $\epsilon$ parameter in Definition \ref{def:apprec}: $\epsilon_{QG}=O\left(e^{-\gamma/G}\right)$. It is interesting to ask when all bulk EFT unitary operators with a well-defined semiclassical description are approximately state-independently reconstructible up to $\epsilon_{QG}$-suppressed errors. When this is possible, we will say that bulk reconstruction is approximately state-independent. 

Consider the set $S_{SC}$ of bulk EFT states admitting a semiclassical description. In the model of non-isometry introduced in \cite{akers2022black}, attention was focused on states and operators constructed as superpositions of a sub-exponential (in $S_{BH}$, i.e. in $1/G$) number of such elements, following the intuition that an arbitrary superposition of an exponential number of semiclassical states need not have a well-defined semiclassical description.\footnote{Note that the fact that these bulk operators have sub-exponential complexity, in the sense that they map sub-exponential states to other sub-exponential states,
does not mean that their boundary reconstruction is also sub-exponential. In fact, when ``simple'' bulk EFT operators mapping between ``simple'' bulk EFT states (e.g. by flipping a single qubit in the bulk) act inside a Python's lunch \cite{brown2020python,Engelhardt:2021mue, Engelhardt:2021qjs}, their boundary reconstruction is conjectured to be exponentially complex.}  Although  not all such ``exponential states'' have a semiclassical description, \textit{some} may. For this reason, we will avoid identifying ``semiclassical'' with ``sub-exponential''.   Thus we will not restrict to sub-exponential superpositions, even though we will be mostly interested in states with interpretations as small excitations around a semiclassical background. 

Now consider a set of states $S_{SI}$ for which the reconstruction of operators mapping between elements is approximately state-independent according to Definition \ref{def:apprec}.  All  states $\ket{\psi_i}\in S_{SI}$ must satisfy $V\ket{\psi_i}=l_i\ket{\Psi_i}$ with $|l_i-l_j|<\epsilon$, because if $\ket{\psi_i},\ket{\psi_j}$  violate this assumption, the bulk operator mapping $\ket{\psi_i}$ to $\ket{\psi_j}$ would not be approximately state-independently reconstructible. Rescaling $V$ then allows us to set all $l_i=1$ up to $\epsilon$-suppressed corrections, i.e., the norms of all states in $S_{SI}$ must be approximately preserved. 
%We will consider states formed by superposition of this set.
%\footnote{For example, the superposition of a small (polynomial in $1/G$) number of semiclassical states with a given geometry generically has a semiclassical description with the same geometry and the state of bulk matter given by the superposition of the states of bulk matter in the original states.} 
The fact that  norms are approximately preserved by $V$ also implies that all overlaps between superpositions of states in the set will be preserved up to $\epsilon$-suppressed corrections. In other words, $S_{SI}$ must be approximately isometrically encoded in $B$ as per the following definition, with $\delta=c\epsilon$ and $c$ an order 1 constant whose precise value is unimportant for our discussion:

\begin{definition}
    Consider a map between two Hilbert spaces $V:A\to B$
    %, where $\log|B|\propto N^2$ and $N\gg 1$ (in holography $N^2\propto 1/G$), 
    and a set of states $S=\{\ket{\psi_i}\}\subseteq A$. $S$ is approximately isometrically encoded in $B$ if the overlap between any two states in $S$ is preserved up to a fixed small error $\delta\ll 1$:
    $|\bra{\psi_i}V^\dagger V \ket{\psi_j}-\langle\psi_i|\psi_j\rangle|<\delta.$
    \label{def:appiso}
\end{definition}
In known examples involving holographic non-isometric maps, overlaps between semiclassical states are approximately preserved under the bulk-to-boundary map up to corrections which behave erratically.  As we discussed in Sec.~\ref{sec:weakvsstrong}, these corrections arise from spacetime wormhole contributions to the gravitational path integral and are (pseudo)randomly distributed (see e.g. \cite{Penington:2019kki,Balasubramanian:2022gmo, Balasubramanian:2022lnw,sasieta2023wormholes,antonini2024holographic,Climent:2024trz, Balasubramanian:2024yxk,cao2023overlapping,penington2020entanglement, almheiri2019entropy,almheiri2020replica,Engelhardt:2021mue, Engelhardt:2021qjs,akers2022black}).   Thus, if we take $|\psi_{i,j}\rangle$ to be orthogonal states in the EFT, and $V|\psi_{i,j}\rangle$ to be their image in the fundamental theory, then we can follow (\ref{eq:mappedinner}) to write
\begin{equation}
    \bra{\psi_i}(V^\dag V-I)\ket{\psi_j}=R_{ij},
    \label{eq:BHrandomoverlaps}
\end{equation}
where $R$ is a $|A|\times |A|$ Hermitian random matrix where each element is independently drawn from a distribution with mean zero and variance $\sigma^2$ ($\overline{R_{ij}}=0$, $\overline{|R_{ij}|^2}=\sigma^2$). Here we denoted the dimension of the bulk Hilbert space by $|A|$.\footnote{Just like in Sec.~\ref{sec:QI}, we will consider the dimension of the bulk and boundary Hilbert spaces to be finite. This can be obtained for instance by restricting to a microcanonical window of the CFT.}

%We will use \eqref{eq:BHrandomoverlaps}
%in the worst-case measures of state-dependence and non-isometry given in Definitions \ref{def:apprec} and \ref{def:appiso}, as opposed to the average measures used in Sec.~\ref{sec:QI}.

The size of the variance $|R_{ij}|^2=\sigma^2$ is determined by  wormhole contributions to the gravitational path integral computing the overlap squared $|\bra{\psi_i}V^\dag V\ket{\psi_j}|^2$ 
and is controlled by the size of the accessible Hilbert space.   For instance, for a single-sided black hole $\sigma^2\propto e^{-S_{BH}}$ \cite{Penington:2019kki} and for a double-sided black hole $\sigma^2\propto e^{-2S_{BH}}$ \cite{Balasubramanian:2022gmo, Balasubramanian:2022lnw}.  
%$\footnote{More precisely, the wormhole contribution gives a Renyi-2 entropy in the exponent.}
In general, such wormholes are associated with the appearance of a non-trivial, compact quantum extremal surface (QES) \cite{Engelhardt:2014gca,Hayden:2018khn,penington2020entanglement, almheiri2019entropy, Penington:2019kki,almheiri2020replica,brown2020python,Engelhardt:2021mue,Engelhardt:2021qjs} separating some bulk region $a$ from the boundary.  The  variance then turns out to be given by $\sigma^2\propto 1/d_{\partial,a}$ where $d_{\partial,a}$ is the dimension of the underlying physical  Hilbert space associated with region $a$. For simplicity, let us restrict to the case where we only consider the encoding of dynamical DOF inside this region $a$, so that $|A|=d_{bulk,a}$ and $|B|=d_{\partial,a}$.
The reason to do this is that states that differ in the region $a$ behind the compact QES are those that receive wormhole contributions to their overlaps, leading to non-isometry. Physics in the simple wedge (outside the QES) is state-independently reconstructible via the HKLL map \cite{Hamilton:2005ju} if we are in the causal wedge, or by some more involved procedure if we are between the boundary of the causal wedge and the QES \cite{Engelhardt:2021mue}. Since we are interested in studying the properties of the state-dependently reconstructible regions/non-isometrically encoded part of the Hilbert space, we can focus our attention only on those, and therefore define the bulk EFT Hilbert space as that of bulk EFT DOF defined in region $a$, and the boundary Hilbert space as that associated with boundary DOF encoding $a$.

For example, for black holes this corresponds to restricting the bulk Hilbert space $A$ to the EFT degrees of freedom behind the horizon, with the fundamental (boundary) Hilbert space $B$ restricted to the physical black hole Hilbert space.  In this case $V:A\to B$ maps the interior  EFT degrees of freedom to fundamental black hole degrees of freedom.  With this restriction, we simply have $\sigma^2=1/|B|$. The inclusion of additional degrees of freedom outside the compact QES (e.g. exterior DOFs in the black hole case) does not significantly change the results of our analysis because they are isometrically encoded/state-independently reconstructible. 

As we have seen, a necessary condition for approximate state-independent bulk reconstruction is that the set $S_{SC}$ of all EFT excitations of semiclassical backgrounds is approximately isometrically encoded in $B$.  Assuming $B$ is large, equation \eqref{eq:BHrandomoverlaps} together with $\sigma^2=1/|B|$ guarantees that each overlap is approximately preserved with high probability. However, without further assumptions on the size of $A$ or the specific probability distribution of the entries of $R$, we cannot conclude that the full set $S_{SC}$ is approximately isometrically encoded with high probability.  This is because, even though each individual entry of the random matrix $R$ in \eqref{eq:BHrandomoverlaps} (i.e., each individual overlap) is small with high probability, it does not follow with high probability that {\it all} entries are small and hence that the full set  $S_{SC}$ is approximately isometrically encoded.  We will discuss this in more detail in Sec.~\ref{sec:strong}. Additionally, as we will see, approximately isometric encoding of $S_{SC}$ is in general necessary but not sufficient for approximate state-independent bulk reconstruction. 

In order to understand the assumptions needed for the entire set $S_{SC}$ to be approximately isometrically encoded with high probability, and those needed for bulk reconstruction to be approximately state-independent, we will separately analyze the cases of weak and strong non-isometry in the sections below.

%\sa{Distinguish two cases. Weak ok, strong not in general, but under reasonable assumptions (cite Akers et al) it is, but this is still not sufficient for approximate state independence, cite 5.1 of Akers et al. Mention somewhere that assumption of random $R$ should probably correspond to random map, but anyway our proof is simpler and only relies on PI results rather than assumptions about the structure of the map. Adapt all proofs with $R$ without rescaling by $B$ and with variance $1/|B|$.}

%Without further assumptions about the size of $A$ and/or the specific distribution the entries of $R$ are drawn from, 

%This is in accordance with the results of \cite{akers2022black}, which suggest that inner products in this context are approximately preserved for semiclassical states with high probability.\footnote{In the model introduced in \cite{akers2022black}, bulk semiclassical states have subexponential complexity, and all subexponential states are approximately isometrically encoded with high probability.} Therefore, the set $S_{SC}$ of semiclassical states satisfies a necessary condition to allow approximate state-independent reconstruction. However, as we will now discuss, this condition is not sufficient. In order to identify under which circumstances $S_{SC}=S_{SI}$ (namely, approximately state-independent reconstruction for all bulk EFT operators is possible on $S_{SC}$), let us distinguish between the two cases of weak and strong non-isometry.

\subsubsection{Weak non-isometry}

First, as already noted, states with semiclassical interpretations $S_{SC}$ form a subset, and not  a subspace, of the full bulk EFT Hilbert space, because superpositions of such states need not be semiclassical themselves.  However, such states can provide an over-complete basis for the complete  Hilbert space, see e.g. \cite{akers2022black,Balasubramanian:2022gmo, Balasubramanian:2022lnw}. We now show that when the bulk-to-boundary map $V:A\to B$ is weakly non-isometric (which requires $|A|\leq |B|$), equation  \eqref{eq:BHrandomoverlaps} with $\sigma^2=1/|B|$ implies that, with high probability, the complete bulk Hilbert space $A$---and thus also the subset of states with semiclassical interpretation---is approximately isometrically encoded in the boundary Hilbert space $B$ according to  Definition \ref{def:appiso}, with $\delta\leq 2\sqrt{|A|/|B|}$:

%First of all, we remark once again that the set of semiclassical states $S_{SC}$ does not in general form a subspace of the full bulk Hilbert space, because exponential states formed by superposition of an exponential number of semiclassical states need not have a semiclassical description. However, it does form a (over)complete basis for the bulk Hilbert space. Assuming that the bulk-to-boundary map $V$ is weakly non-isometric (namely $|A|<|B|$), we will now show that equation \eqref{eq:BHrandomoverlaps} with $\sigma^2=1/|B|$ implies that, with high probability, the full bulk Hilbert space $A$ is approximately isometrically encoded in the boundary Hilbert space $B$ according to Definition \ref{def:appiso}, with $\delta\leq 2\sqrt{|A|/|B|}$. 

%More precisely, we will prove the following general proposition:

\begin{proposition}
    Let $V:A\to B$ be a weakly non-isometric map between two Hilbert spaces, with $1\ll |A|\leq |B|^\alpha\ll |B|$ for some $\alpha<1$. Assume that a complete orthonormal basis $\{\ket{\psi_i}\}$ of $A$ satisfies
    \begin{equation}
        \bra{\psi_i}(V^\dag V-I)\ket{\psi_j}=R_{ij},
    \end{equation}
    where $R$ is an $|A|\times |A|$ Hermitian random matrix where each element is independently drawn from a distribution with mean zero and variance $1/|B|$ ($\overline{R_{ij}}=0$, $\overline{|R_{ij}|^2}=1/|B|$). Then, with probability $1 - O(e^{-\gamma |A|})$ for some $O(1)$ constant $\gamma$, the full Hilbert space $A$ is approximately isometrically encoded in $B$ following  Definition \ref{def:appiso} with 
    $$\delta= 2\sqrt{\frac{|A|}{|B|}}.$$
    \label{propiso}
\end{proposition}

\begin{proof}
    Consider two arbitrary (not necessarily semiclassical) normalized states $\ket{\psi_\alpha},\ket{\psi_\beta}\in A$. Then we can use the Cauchy-Schwarz inequality and the definition $R=V^\dagger V-I$ to write\footnote{We thank Luca Iliesiu, Guanda Lin, and Pratik Rath for discussions related to this bound.}
\begin{equation}
%\begin{aligned}
    |\bra{\psi_\alpha}(V^\dag V-I)\ket{\psi_\beta}|\leq \sqrt{\langle\psi_\alpha|\psi_\alpha\rangle}\sqrt{\langle\psi_\beta|R^\dagger R|\psi_\beta\rangle} \leq \sigma_{max}
 %   \end{aligned}
\end{equation}
where  $\sigma_{max}$ is the largest singular value of $R$.  Note that, since $R$ is Hermitian and the singular values of a Hermitian matrix are equal to the absolute values of its eigenvalues $\lambda_i$,  $\sigma_{max}=\max_i|\lambda_i|\equiv \lambda_{max}$.

%First, we can use Wigner's semicircle law \cite{Wigner1955CharacteristicVO,Anderson_Guionnet_Zeitouni_2009} to  obtain the expected value of $\lambda_{max}$.

We therefore want to bound $\lambda_{max}$.  Given  $R$ with zero mean and variance $\sigma^2=1/|B|$, define $\tilde{R}=R/\sqrt{|A|}$. The absolute value of the largest magnitude eigenvalue of $\tilde{R}$ is then $\tilde{\lambda}_{max}=\lambda_{max}/\sqrt{|A|}$. Wigner's semicircle law \cite{Wigner1955CharacteristicVO,Anderson_Guionnet_Zeitouni_2009} tells us that $\tilde{\lambda}_{max}=2\sigma=2/\sqrt{|B|}$ when  $|A|\to\infty$. Therefore in the large $A$ limit, the maximum eigenvalue of $R$ is $\lambda_{max}=2\sqrt{|A|/|B|}$ up to subleading corrections suppressed by $1/|A|$ \cite{goetze,feier2012methods}.\footnote{The same result can be obtained by assuming the bulk-to-boundary map $V$ to be random and using the Marchenko-Pastur distribution for the eigenvalues of large rectangular matrices \cite{Marčenko_1967}. We thank Pratik Rath for discussions on this point.} In fact we will show below that $\lambda_{max} \to 2\sqrt{|A|/|B|}$ with probability $1$ in the large $A$ limit.

%{\color{red} I have not read the latter two papers.   Isn't it much stronger than an expected value as was originally stated in the text?  I would have thought the at large $|B|$, we approach this value with probability 1 and this is what we show below? I edited to reflect this.}

The Tracy-Widom distribution \cite{tracywidom} gives the probability of deviating from this value by some  amount.  First assume that the eigenvalue of largest magnitude is positive. Then, consider the probability that the largest eigenvalue of a matrix drawn from a  $N\times N$ random matrix ensemble is less that $x$: $\mathbb{P}_{N,\beta}(\lambda_{max} < x)$, where $\beta=1,2,4$ indicates the orthogonal, unitary and symplectic ensembles.  We are interested in the unitary ensemble ($\beta=2$) since $R$ is Hermitian.  For ensembles with standard deviation $\sigma$ and matrices of size $N$, Tracy and Widom found it convenient to consider $\lim_{N\to\infty} \mathbb{P}( \lambda_{max}/\sigma - 2\sqrt{N} \leq x/N^{1/6} )$. The left hand side of the inequality measures the deviation of the maximum eigenvalue in units of the ensemble standard deviation from the large $N$ limit dictated by the Wigner semicircle distribution. The $N^{1/6}$ in the denominator on the right hand side accounts for the known scaling of this deviation.  For us, $N= |A|$ and $\sigma=1/\sqrt{|B|}$, so we define
\begin{equation}
    F_\beta(x)\equiv\lim_{|A|\to\infty}\mathbb{P}\left(\lambda_{max}-2\sqrt{\frac{|A|}{|B|}}\leq \frac{x}{|A|^{1/6}\sqrt{|B|}}\right) \,.
    \label{eq:tracywidom1}
\end{equation}
$F_\beta(x)$ is the probability that the maximum eigenvalue deviates from the large $N$ limit by less than the amount on the right side where the known scalings with matrix size and ensemble standard deviation have been accounted for. We are therefore interested in studying $1-F_\beta(x)$, namely the probability of a deviation from the large $N$ limit by at least the amount on the right hand side of \eqref{eq:tracywidom1}. Tracy and Widom showed that
\begin{equation}
1-F_\beta(x)\sim e^{-|x|^{3/2}}
\label{eq:tracywidom2}
\end{equation}
for positive $|x|\gg 1$ and decays even faster for large negative $x$ \cite{tracywidom}. A similar result can be obtained if the eigenvalue with largest magnitude $\lambda_{max}$ is negative. These results were initially obtained for Gaussian ensembles \cite{tracywidom}, but apply  more generally \cite{tracywidom-universal1,tracywidom-universal2,Soshnikov1999UniversalityAT}.

%and the index $\beta$ labels the specific matrix ensemble considered. {\color{red} First, what is $x$ in 3.11 and what is $F_\beta$ -- is it probability of a large deviation of some kind?  We should explain! Second, $\beta$ does not appear on the right hand side, so why are we including it on the left hand side.  In any case, isn't $R$ Hermitian, so don't we want an the ensemble of Hermitian matrices, i.e. the GUE ensemble? Third, what does the symbold $\mathbb{P}$ mean in this equation?  Fourth why are we taking the large $A$ limit? Is the quantity within the parentheses the Tracy-Widom result which we take a limit on to get the answer relevant for us, or is their result that for large matrices we should take this to get their result.  The limit is also not well defined as stated becuase the LHS of the inequality becomes minus infinity.}

To obtain the probability of a leading order, i.e., $O(\sqrt{|A|/|B|})$, deviation from the large $|A|$ limit, we can can set $x = \Delta |A|^{2/3} \gg 1$ in (\ref{eq:tracywidom1}), (\ref{eq:tracywidom2}), where $\Delta$ is any positive $O(1)$ number. This gives\footnote{Technically, this result holds in the asymptotic limit $|A|\to\infty$, but it can be extended for large but finite $|A|$ \cite{tracywidom-finiten}.}
\begin{equation}
    \mathbb{P}\left(\left|\lambda_{max}-2\sqrt{\frac{|A|}{|B|}}\right|\gtrsim \Delta\sqrt{\frac{|A|}{|B|}}\right)\sim e^{-\Delta^{3/2}|A|} \, .
\end{equation}
Therefore, with probability exponentially  close to 1, $\lambda_{max}=2\sqrt{|A|/|B|}$ up to subleading corrections. Putting everyting together, for any two states $\ket{\psi_\alpha},\ket{\psi_\beta}\in A$ 
\begin{equation}
%\begin{aligned}
    |\bra{\psi_\alpha}(V^\dag V-I)\ket{\psi_\beta}|\leq \delta
 %   \end{aligned}
\end{equation}
with $\delta=2\sqrt{|A|/|B|}\ll 1$. This concludes our proof that the full Hilbert space $A$ is approximately isometrically encoded in $B$ with probability exponentially close to 1.
\end{proof}

In Proposition \ref{propiso} we have assumed $|A|\gg 1$, which allowed us to use known results for large random matrices. But what if $|B|\gg |A|=O(1)$? This situation can arise for example in Python's lunch geometries \cite{brown2020python,Engelhardt:2021mue, Engelhardt:2021qjs}\footnote{Here and in the following we are referring to bag-of-gold-type geometries, for which the whole dual theory has a Python's lunch, rather than settings involving Python's lunches for boundary subregions.} when we only consider few bulk EFT degrees of freedom inside the lunch. 
%In this case, without further assumptions on the distribution of the entries of $R$, we can obtain approximate isometric encoding for $A$ only with probability $1/|B|$. In fact, we have 
To study this case, first observe that
\begin{equation}
    \lambda_{max}\leq ||R||_1=\max_{j}\sum_{i=1}^{|A|}|R_{ij}|\leq |A|(\max_j\max_i |R_{ij}|),
\end{equation}
where we used the fact that $\lambda_{max}$ is a lower bound for any matrix norm, and in particular for the matrix 1-norm. We can now ask what is the probability that the largest entry of $R$ is larger than some inverse power of $|B|$. From Chebyshev's inequality, and the fact that the standard deviation of the random matrix ensemble is $\sigma=1/\sqrt{|B|}$, the probability that a given of $R$ has a magnitude larger that $1/|B|^\zeta$ for $\zeta>0$ is
\begin{equation}
    \mathbb{P}\left(|R_{ij}|\geq \frac{1}{|B|^\zeta}\right)\leq |B|^{2\zeta-1} \, .
\end{equation}
%
%for a given $\zeta>0$. 
Therefore, the probability that at least one matrix element is larger than $|B|^{-\zeta}$ is given by
\begin{equation}
    \mathbb{P}\left(\max_{ij}|R_{ij}|\geq \frac{1}{|B|^\zeta}\right)\leq 1-\left(1-|B|^{2\zeta-1}\right)^{\frac{|A|(|A|+1)}{2}}\approx \frac{|A|(|A|+1)}{2}|B|^{2\zeta-1}
    \label{eq:probche}
\end{equation}
where  the exponent in the first inequality counts the independent entries of $R$, and we assumed $0<\zeta<1/2$ so that $|B|^{2\zeta-1}\ll 1$. This implies that
\begin{equation}
    \mathbb{P}\left(\lambda_{max}\leq \frac{|A|}{|B|^\zeta}\right)\geq 1-\frac{|A|(|A|+1)}{2}|B|^{2\zeta-1}
\end{equation}
and $A$ is approximately isometrically encoded in $B$ with high probability up to corrections $\delta=O\left(1/|B|^\zeta\right)$ with $0<\zeta<1/2$.

This means that if $|B|=O\left(\exp^{\gamma/G}\right)$, as occurs in the Python's lunch cases mentioned above, the probability of approximately isometric encoding of $A$ with $
\delta=O\left(\exp^{-\gamma\zeta/G}\right)$ is $1-O\left(\exp^{-\gamma(2\zeta-1)/G}\right)$. Second, Chebyshev's inequality generally gives a very loose bound on the probability and under some reasonable assumptions about the nature of the ensemble, the bound \eqref{eq:probche} can be made stricter. For instance, by assuming a Gaussian distribution for the entries of $R$, the probability of approximate isometric encoding approaches 1 exponentially in $|B|$ rather than polynomially. This also allows us to make $\zeta$ larger than $1/2$, obtaining a stricter bound on the corrections to the overlap with very high probability. We can therefore conclude that in this case too the whole Hilbert space $A$ is approximately isometrically encoded in $B$ with very high probability.

So far we have shown that, if the map $V$ is weakly non-isometric, it is enough to have a complete basis of (semiclassically well-defined) states satisfying \eqref{eq:BHrandomoverlaps} to ensure that the whole bulk EFT Hilbert space $A$ is approximately isometrically encoded in $B$, i.e., $V$ is an approximate isometry. This naturally implies that the set of semiclassical states $S_{SC}$ is also approximately isometrically encoded. We can now ask whether all bulk EFT operators mapping between semiclassical states can be approximately state-independently reconstructed. We will show below that the approximately isometric encoding of all of $A$ turns out to be a sufficient condition for all operators defined on $A$  to be approximately state-independently reconstructible on $B$ according to Definition \ref{def:apprec}.

\begin{proposition}
\label{propstate}
    Let $V:A\to B$ be an approximately isometric map between two Hilbert spaces, i.e.
    \begin{equation}
        |\langle\psi_i|(V^\dagger V-I)|\psi_j\rangle|\leq \delta \quad \quad \forall \ket{\psi_i},\ket{\psi_j}\in A.
        \label{eq:isoassumption}
    \end{equation}
    For any unitary operator $O$ acting on $A$ define the unitary operator 
    \begin{equation}
        \tilde{O}=WOW^\dagger
    \end{equation}
    acting on $B$, where $W=UV_0Q$, $V_0$ is a $|B|\times |A|$ rectangular diagonal matrix with all 1's on the main diagonal, and $U,Q$ are the unitary matrices yielding the singular value decomposition of $V$. Then $\tilde{O}$ is an approximately state-independent, unitary reconstruction of $O$ according to Definition \ref{def:apprec} with
    \begin{equation}
        \epsilon=\delta\sqrt{1+\delta}+2\left(1+\delta\right)\left(\sqrt{1+\delta}-1\right)+\sqrt{1+\delta}\left(\sqrt{1+\delta}-1\right)^2\approx 2\delta
    \end{equation}
    with the last approximation valid for $\delta\ll 1$. 
\end{proposition}

\begin{proof}
    For any operator $O$ acting on $A$, we are interested in bounding the spectral norm of $\tilde{O}V-VO$:
    \begin{equation}
        ||\tilde{O}V-VO||_2\equiv\max_{\ket{\psi}\in A}||(\tilde{O}V-VO)\ket{\psi}|| \, ,
    \end{equation}
    and we recall that the spectral norm of a matrix $M$ is equal to the largest singular value of $M$, $||M||_2=\sigma_{max}^{(M)}$.  In what follows all states are normalized to 1. 
    
    First, consider the singular value decomposition of $V$
    \begin{equation}
        V=UV_DQ=U\left(V_0+V_1\right)Q
    \end{equation}
    where $V_0$ is a rectangular diagonal matrix with all 1's on the main diagonal and $V_1=V_D-V_0$. Then define $V=W+\Delta$, with $W=UV_0Q$ and $\Delta=UV_1Q$. Note that all singular values of $W$ are 1 by construction  ($\sigma_i^{(W)}=1 \quad \forall i=1,...,|A|$) and therefore $W$ is isometric $W^\dag W=I$. Given the assumption \eqref{eq:isoassumption}, we can bound the largest singular value of $V$:
    \begin{equation}
        \left(\sigma_{max}^{(V)}\right)^2-1=\max_{\ket{\psi}\in A}\langle\psi|V^\dag V|\psi\rangle-1\leq \max_{\ket{\psi}\in A}|\langle\psi|(V^\dag V-I)|\psi\rangle|\equiv \lambda_{max}^{(V^\dag V-I)}\leq \delta,
        \label{eq:boundV}
    \end{equation}
    where $\lambda_{max}^{(V^\dag V-I)}$ labels the absolute value of the largest magnitude eigenvalue of $(V^\dag V-I)$. Equation \eqref{eq:boundV} immediately implies $||V||_2=\sigma_{max}^{(V)}\leq \sqrt{1+\delta}$. Because by construction $\sigma_i^{(\Delta)}=\sigma_i^{(V)}-1$, this also implies the bound $||\Delta||_2=\sigma_{max}^{(\Delta)}\leq \sqrt{1+\delta}-1$.

    Now consider the operator $\tilde{O}=WOW^\dag$ on $B$. Note that since $O$ is unitary by assumption and $W$ is isometric, $\tilde{O}$ is unitary. We can now bound the spectral norm of interest:
    \begin{equation}
    \begin{aligned}
        ||\tilde{O}V-VO||_2&=||\left(V-\Delta\right)O\left(V^\dag-\Delta^\dag \right)V-VO||_2\\
        &\leq ||VO\left(V^\dag V-I\right)||_2+||\Delta O V^\dag V||_2+||VO\Delta^\dag V||_2+||\Delta O \Delta^\dagger V||_2\\
        &\leq ||V||_2||V^\dag V-I||_2+2||\Delta||_2||V||_2^2+||\Delta||_2^2||V||_2
        \end{aligned}
    \end{equation}
    where we used the definition of $W$, the triangle inequality, the submultiplicative property of the norm $||AB||_2\leq ||A||_2||B||_2$, $||O||_2=1$, and $||M||_2=||M^\dag||_2$.\footnote{It is easy to check that the non-zero singular values of a matrix (and therefore the largest one) are equal to those of its conjugate.} We have already bounded $||\Delta||_2$ and $||V||_2$. Moreover, in equation \eqref{eq:boundV} we saw that the absolute value of the largest magnitude eigenvalue of $V^\dag V-I$ is bounded as $\lambda_{max}^{(V^\dag V-I)}\leq \delta$. Notice that $V^\dag V-I$ is real and symmetric, and therefore $\lambda_{max}^{(V^\dag V-I)}=\sigma_{max}^{(V^\dag V-I)}=||V^\dag V-I||_2\leq \delta$. Putting together all the bounds on the norms, we finally obtain
    \begin{equation}
        ||\tilde{O}V-VO||_2\leq \epsilon\equiv \delta\sqrt{1+\delta}+2\left(1+\delta\right)\left(\sqrt{1+\delta}-1\right)+\sqrt{1+\delta}\left(\sqrt{1+\delta}-1\right)^2\approx 2\delta
    \end{equation}
    where in the last approximation we assumed $\delta\ll 1$ and kept only the leading order in $\delta$.
\end{proof}

To summarize, we showed that weak non-isometry and the existence of a basis of (semiclassical) states satisfying equation \eqref{eq:BHrandomoverlaps} with $\sigma^2=O(1/|B|)$ and $|B|\gg 1$ are together sufficient to ensure that the bulk-to-boundary map $V:A\to B$ is an approximate isometry, and that all unitary operators on $A$ can be approximately state-independently reconstructed by unitary operators on $B$.

To interpret these  results within the holographic framework, let us start by considering a weakly non-isometric bulk-to-boundary map with $|B|=O\left(e^{\gamma/G}\right)$ and $|A|\leq |B|^\alpha$ with $\alpha<1$. Examples involving black holes with this kind of encoding map were described in Sec. \ref{sec:weakvsstrong}. Our results then indicate that the whole bulk EFT Hilbert space $A$ is approximately isometrically encoded in the fundamental Hilbert space $B$ up to $e^{-\gamma/G}$-suppressed corrections with probability approaching 1 at least exponentially fast in $1/G$.\footnote{This statement does not make any assumptions about the size of $|A|$ or the probability distribution of corrections to the EFT overlaps. If $|A|$ is also exponential in $1/G$, Proposition \ref{propiso} guarantees that the probability of approximately isometric encoding is doubly-exponentially close to 1. The same result can be also obtained for $|A|=O(1)$ under assumptions on the probability distribution of overlaps that were discussed above.} Similarly, all operators acting on $A$ are approximately state-independently reconstructible up to corrections suppressed by $\epsilon_{QG}=O\left(e^{-\gamma/G}\right)$. Hence, the set of bulk unitaries mapping between semiclassical states can also be reconstructed in an approximately state-independent manner. Therefore, we can conclude that bulk reconstruction of the EFT is approximately state-independent in such holographic settings. Since  the  reconstruction  of the EFT (with its naive inner product that ignores quantum gravity corrections) is always approximate up to $\epsilon_{QG}$-suppressed errors, our reconstruction is as exact as it can be. 
%We can then conclude that in this case bulk reconstruction is simply state-independent at any attainable level of accuracy.  {\color{red} Here I have the same confusion that I mentioned earlier.  There is an exact underlying duality and an associated eact reconstruction.  Is this supposed to be saying that a certain kind of ``semiclassical'' reconstruction via HKLL or something is approximate?  I think we should clarify a bit what we want to mean by bulk reconstruction.}

%Note that, as we have already mentioned, the case we just analyzed is exactly the situation faced when considering 
In Python's lunch spacetimes, the bulk EFT Hilbert space of the lunch $A$ (i.e., the code subspace) is chosen to be much smaller than the corresponding fundamental Hilbert space $B$, whose size is set by $\exp\left(\mathcal{A}(\gamma_{aptz})/4G\right)$, where $\gamma_{aptz}$ is the ``appetizer'' non-minimal QES. The existence of $\gamma_{aptz}$ is associated with the presence of wormhole contributions to the overlaps between different semiclassical states of the Python, signaling a weakly non-isometric map.\footnote{We thank Arvin Shahbazi-Moghaddam for discussions on this point.} This can also be concluded from the use of postselection in the bulk-to-boundary map in tensor network toy-models of Python's lunches  \cite{brown2020python}. Notice that in the setups studied in  \cite{brown2020python,Engelhardt:2021mue,Engelhardt:2021qjs,Engelhardt:2023bpv}, $\gamma_{aptz}$ is a geometric surface with area $O(G^0)$ and the correction to overlaps is therefore $e^{-\gamma/G}$-suppressed.  As we have seen, this implies that the whole bulk Hilbert space is approximately isometrically encoded, which in turns implies that reconstruction is approximately state-independent for all bulk operators up to $\epsilon_{QG}$-suppressed errors. This is so even though the reconstruction of even simple operators acting in the Python's lunch is exponentially complex \cite{brown2020python,Engelhardt:2021mue,Engelhardt:2021qjs,Engelhardt:2023bpv}. 

%\sa{This might need to be modified. I think not even state-dependent reconstruction can work in this case. If overlaps (and norms as well) are only polynomially suppressed, even for semiclassical states, then operators mapping between different states cannot be unitarily reconstructed with exponential precision not even state-dependently, see arguments at beginning of 3.2. What are the implications of this? Probably that there are some perfectly fine semiclassical ops in the bulk that just can't be reconstructed unitarily precisely, and so basically don't make sense in the boundary. I think this could be a solid sign of EFT breakdown. Notice that it does not require a strong non-isometry like the old BH one, it's the same even for weak non-isometry like the closed universes setups. This seems to point to the fact that these semiclassical theories just don't make sense in holography (if there's a sense in which they make sense, it's very non-standard (like options c and d of essay with Pratikx)'').}
 What happens when $|A|\ll |B|$ and $|B|$ is sub-exponential in $1/G$? Then $\delta$ and $\epsilon$ appearing in Propositions \ref{propiso} and \ref{propstate} are at best polynomially suppressed in $G$. One example of this situation, for which $|B|=G^0$, can be obtained in the setting involving closed universes discussed in \cite{Antonini:2023hdh} (for a sufficiently small bulk Hilbert space $A$ and a sufficiently large entanglement between the closed universe and the two auxiliary AdS spacetimes). In this case, the approximately state-independent reconstruction of bulk operators we described is only polynomially precise. Can a state-specific reconstruction do better? From equation \eqref{eq:differentscalings}, we can conclude that the answer is no. Bulk operators mapping between states whose norms differ polynomially in $G$ under the map $V$ can only be unitarily reconstructed up to errors polynomial in $G$. The only way to obtain an exponentially precise bulk reconstruction would be to either allow the boundary reconstruction of bulk unitaries to not be unitary, or modify the bulk-to-boundary map by introducing a state normalization after the action of the map. Clearly, both options are undesirable: the former would imply that bulk physical operations do not correspond to boundary physical operations; the latter would lead to a state-dependent and non-linear bulk-to-boundary map. We are therefore led to conclude that some apparently well-defined semiclassical EFT operations in the bulk simply cannot be implemented in the fundamental description beyond polynomial (in $G$) precision. In setups similar to that described in \cite{Antonini:2023hdh}, the errors are $O(G^0)$. This is suggestive of a breakdown of bulk EFT, which is not a reliable low-energy description of the fundamental theory at all orders in $G$ \cite{antonini2024holographic}.

It is also interesting to ask is what happens when $V$ is weakly non-isometric, the set of semiclassical states is approximately isometrically encoded up to $O\left(e^{-\gamma/G}\right)$-suppressed corrections, $|B|=O\left(e^{-\gamma/G}\right)$, and $|A|/|B|=O(G^\beta)$ for some $\beta\geq 0$.   In this case, $\delta$ and $\epsilon$ appearing in Propositions \ref{propiso} and \ref{propstate} are again polynomially suppressed in $G$, and our bounds do not guarantee that the approximately state-independent reconstruction of bulk operators is exponentially precise for generic operators. However, it seems plausible that all unitary operators mapping between semiclassical states are approximately state-independently reconstructible up to $\epsilon_{QG}$-suppressed errors on the subset of semiclassical states, and therefore bulk reconstruction is approximately state-independent. We leave a further exploration of this interesting possibility to future work.

Finally, a similar result was obtained in the context of the Haar random map model of \cite{akers2022black} (see Sec.~5.4). In that case the map is strongly non-isometric, but a weakly non-isometric map arises by restriction of the domain to a subspace $C \subset A$ such that $|C|=|B|^\gamma$ with $\gamma<1$. Assuming $|B|$ is large and using the properties of the Haar random map, the authors showed that $C$ is approximately isometrically encoded in $B$. Consequently, all operators acting on $C$ can be approximately state-independently reconstructed with high probability. A first difference in our argument is that we did not make specific assumptions about the random nature of the bulk-to-boundary map $V$, but rather directly employed the random behavior of overlaps between semiclassical states signaled by wormhole corrections in the gravitational path integral. This allowed us to obtain a more immediate proof of our statements which applies to generic gravitational setups displaying this behavior. A second difference is that our boundary operators $\tilde{O}$ reconstructing bulk unitaries $O$ are exact unitaries, whereas in \cite{akers2022black} they are approximate unitaries built using the approximately isometric map. Since we would like to construct physical boundary observables corresponding to physical bulk observables, having a (physical) exact unitary seems preferrable.\footnote{Naturally, this difference is subtle because it seems plausible that exact unitaries exist within the construction of \cite{akers2022black} which approximate the approximate unitaries well enough to not alter the main result.}

\subsubsection{Strong non-isometry}
\label{sec:strong}
%\sa{Keep going}

%\sa{Explain that now semiclassical states need not be approximately isometrically encoded, but they are under reasonable assumptions about the distribution (measure concentration results), and this is in accordance with the results of \cite{akers2022black}, which suggest that inner products in this context are approximately preserved for semiclassical states with high probability.\footnote{In the model introduced in \cite{akers2022black}, bulk semiclassical states have subexponential complexity, and all subexponential states are approximately isometrically encoded with high probability.}}

If $V$ is strongly non-isometric and thus $|A|>|B|$, clearly not all of $A$ can be approximately isometrically encoded in $B$. In fact, for any state $\ket{\psi_N}$ in the kernel of $V$ we have $|\bra{\psi_N}(V^\dag V-I)\ket{\psi_N}|=1$. In other words, $V$ is not an approximate isometry in this case. Therefore, there can be no approximate state-independent reconstruction for all operators on $A$. On the other hand, since we are only interested in semiclassical EFT states around a fixed background and unitary operators mapping between them, one might hope that restricting to the subset $S_{SC}$ of semiclassical states and to the appropriate bulk EFT operators mapping between them, an approximate state-independent reconstruction for all such operators might exist.  In fact, we will show that we are not guaranteed, without further assumptions, that all semiclassical states are approximately isometrically encoded, and that even if we make those assumptions and $S_{SC}$ is approximately isometrically encoded,  the reconstruction of operators still cannot be approximately state-independent.

First of all, notice that now we are not guaranteed that the set of all semiclassical states $S_{SC}$ is approximately isometrically encoded with high probability, and therefore satisfies the necessary condition for approximate state-independent reconstruction. To obtain such a result, we would need to make some assumption about the size of $S_{SC}$ and the specific probability distribution governing the corrections to the overlaps between semiclassical states \eqref{eq:BHrandomoverlaps}. This was achieved for example in \cite{akers2022black} by a) only considering states that are a sub-exponential superpositions of a basis of semiclassical states, b) bounding the size of the set of all such sub-exponential states, and c) using measure concentration results for the specific Haar random map considered there.

Let us assume that in our case, under analogous reasonable assumptions, the set of all semiclassical states $S_{SC}$ is similarly approximately isometrically encoded. Then, the results of \cite{akers2022black} suggest that reconstruction of bulk EFT operators mapping between them is nonetheless state-dependent (in particular, it is not even approximately state-independent). In fact, they imply that\footnote{Here we are adapting the results of \cite{akers2022black} to our setup and notation. For the exact result see Theorem 5.1 and Appendix H in \cite{akers2022black}.} if a) $\log |A|$ is sub-exponential in $\log |B|$ and b) all sub-exponential (in $\log|B|$) operators are approximately state-independently reconstructible on all sub-exponential states with errors $\epsilon=O\left(1/|B|^\gamma\right)$ for some $\gamma>0$, then $|A|/|B|\leq 1+\eta$, with $\eta=O\left(1/|B|^\gamma\right)$. In other words, the kernel of $V$ must be (almost) empty. Although, as we have discussed, the identification between sub-exponential states/operators and semiclassical states/operators can be subtle, the proof of this result given in \cite{akers2022black} relies only on the existence of an approximately isometrically encoded basis of states for $A$ (which in \cite{akers2022black} was taken to be a basis of sub-exponential states).\footnote{We thank Geoff Penington for a discussion on this point.} Since a basis of semiclassical states for the bulk Hilbert space $A$ certainly exists,\footnote{In fact, the bulk EFT Hilbert space is defined as the span of a basis of semiclassical states, e.g. different states for bulk matter on top a fixed geometry, or different geometric microstates as explained in Sec. \ref{sec:weakvsstrong}. } the same result can be obtained for semiclassical states and bulk EFT operators mapping between them under the assumption that $S_{SC}$ is approximately isometrically encoded. 
%\sa{It would be nice to verify carefully that the proof of Theorem 5.1 of \cite{akers2022black} really only relies on this, and maybe make this argument more precise with a calculation. Also, I'm claiming here that a basis of semiclassical states certainly exists (I used this also above in this section). The intuition I have is that we define the bulk EFT by putting DOF, say qubits, on top of a fixed geometry. So all states differing from a semiclassical state by a flip of an $O(1)$ number of qubits should also be semiclassical. This implies that the computational basis is made of semiclassical states, and therefore a basis of semiclassical states exists. Do we agree? Do we have a more precise way to claim this? Should we spell it out or is it obvious?} 
Therefore, the set of semiclassical bulk EFT operators cannot all be approximately state-independently reconstructed on $S_{SC}$ if $V$ is strongly non-isometric. This also implies that approximate isometric encoding (in this case of $S_{SC}$) is necessary but not sufficient for approximate state-independence.

As a result, if we insist on keeping all the possible semiclassical states in $A$, bulk reconstruction is necessarily state-dependent in the strongly non-isometric case. As it was already pointed out in \cite{akers2022black}, the only way to obtain approximate state-independent bulk reconstruction is to restrict our attention to a subspace of $A_\alpha \subset A$ of dimension at most $|B|^{\alpha}$ with $\alpha<1$. This way, we simply reduce to the weakly non-isometric case analyzed above, and we obtain approximate state-independent reconstruction for all operators acting on the subspace $A_\alpha$. This result is also in complete accordance with the results of \cite{Hayden:2018khn} in the presence of black holes, which suggest that approximately state-independent bulk reconstruction should work for any subspace of $A$ of size $e^{\alpha S_{BH}}$ with $\alpha<1$, where $|B|=e^{S_{BH}}$.

However, because semiclassical states are a (over)complete basis for $A$, this procedure clearly also implies the restriction to a subset $S_0\subset S_{SC}$ of semiclassical states as well as a restriction on the allowed EFT operators. In particular any bulk EFT operator mapping a state in $S_0$ to a state in $S_0^\perp \cap S_{SC}$, which seems to be perfectly well-defined from a bulk semiclassical EFT point of view, must be deemed inadmissible. This approach seems rather unphysical from a bulk EFT point of view, especially given the arbitrariness in the choice of subspace. For this reason, it seems reasonable to conclude that, in the presence of a strongly non-isometric map, one should simply keep the whole bulk EFT Hilbert space $A$ and accept the presence of state-dependent bulk reconstruction. This is the point of view adopted by the authors of \cite{akers2022black}.
Although it currently is the best-understood and a physically well-motivated option, embracing strong non-isometry and state-dependence is not free from issues \cite{bousso2020unitarity,bousso2022islands, antonini2024holographic}. For instance, in the context of black hole evaporation, corrections to the overlaps between semiclassical states become polynomial in $G$ when $\mathcal{A}_{hor}=O(G)$  but still $|B|\gg 1$, and we incur in the EFT breakdown we already discussed in the weakly non-isometric case. Note that this happens before the black hole itself becomes Planckian, and thus before EFT is expected to break down. For instance, the corrections are already polynomial in $G$ when $\mathcal{A}=O\left(G\log \frac{1}{G}\right)$ and therefore much larger than the Planck scale.

\paragraph{Summary: } Before moving on, let us summarize the results of this subsection. We found that if the overlap between any two semiclassical states is approximately preserved by the bulk-to-boundary map (as suggested by gravitational path integral computations) and the map is weakly non-isometric, then the map is an approximate isometry and the reconstruction of bulk operators is approximately state-independent. In particular, if the corrections to the overlaps are suppressed by $e^{-\gamma/G}$, this approximate reconstruction is as accurate as it can be, because quantum gravity effects (which are unaccounted for in this discussion) will also give contributions at the same order. On the other hand, if the map is strongly non-isometric, bulk EFT operators can only be reconstructed state-dependently, even when restricting to bulk semiclassical operators and bulk semiclassical states. Finally, regardless of whether the map is weakly or strongly non-isometric, if the corrections to the inner products between semiclassical states are only polynomially suppressed in $G$---as it is the case e.g. for the closed universe setup of \cite{Antonini:2023hdh} and for evaporating black holes at late times, when $\mathcal{A}_{hor}=O(G)$---then bulk EFT operators can only be reconstructed to polynomial precision $G$. This is a sign that the bulk EFT is not actually valid to all orders in $G$ as one might have naively expected.\footnote{Notice that in some examples of this class of setups, radically different semiclassical bulk descriptions corresponding to the same boundary state can be constructed \cite{bousso2020unitarity,bousso2022islands, antonini2024holographic}. This also points to an ambiguity in the holographic dictionary, which could be related to the breakdown of EFT \cite{antonini2024holographic}.} We leave further exploration of the implications of this result to future work.

\subsection{Non-isometry, state-dependence, and causal connectivity}
\label{sec:causal}

In many known examples in holography, the existence of a global causal horizon, i.e., the boundary of the intersection between the bulk causal past and the bulk causal future of the whole boundary,\footnote{The bulk EFT Hilbert space is only defined perturbatively in $G$, and we are considering perturbative bulk-to-boundary maps. Thus, we do not consider non-perturbative effects such as a change of the global causal structure of a spacetime due to the complete evaporation of a black hole. Therefore, for our purposes a black hole event horizon is global.} or a disconnected patch of the gravitational spacetime (i.e., a closed universe) is strictly related to the presence of a non-isometric global bulk-to-boundary map and hence state-dependent reconstruction. In this subsection we will try to sharpen this relationship.  We will first argue that non-isometry implies the existence of a bulk region which is causally disconnected from the boundary where the dual theory is defined. Then we will conjecture that the existence of a global horizon implies non-isometry and state-dependence under reasonable assumptions about the bulk EFT Hilbert space. Our goal is to clarify what we know about the relationship between non-isometry, state-dependence, and causal connectivity, and to set out questions that still need to be answered.

The first implication is easy to argue for. If a holographic global bulk-to-boundary map $V:A\to B$ is non-isometric, the results of Sec.~\ref{sec:QI} tell us that there must exist a subregion of the bulk spacetime for which global bulk reconstruction is necessarily state-dependent. But we know that reconstruction in the domain of dependence of the causal wedge is state-independent and can be explicitly obtained via the HKLL dictionary\footnote{Note that in the presence of a non-minimal QES, state-independent reconstruction can be extended beyond the domain of dependence of the causal wedge to include the entire ``simple wedge'' \cite{Engelhardt:2021mue}. Therefore, the state-dependently reconstructible region in this case must be outside the simple wedge, and therefore again causally disconnected from the boundary.} \cite{Balasubramanian:1998sn,Balasubramanian:1999ri,Hamilton:2005ju,hamilton2006holographic,hamilton2007local,hamilton2008local}. This automatically implies that the region that requires state-dependent reconstruction is outside the domain of dependence of the causal wedge of the whole dual theory for at least some states in the bulk EFT Hilbert space.\footnote{Recall that we are including in the definition of ``dual theory'' all degrees of freedom present in the fundamental description, including any non-gravitational bath coupled to the holographic theory.} 
Since the holographic dual theory lives on the AdS boundary,\footnote{The bulk EFT representation of any non-gravitational auxiliary bath (such as that considered in  \cite{almheiri2020replica}) coupled to the holographic theory living on the boundary is outside the gravitational bulk spacetime, and coupled to it by transparent boundary conditions.} the region  must be outside the domain of dependence of the causal wedge of the whole boundary, namely it must be causally disconnected from the whole boundary. Clearly, this is  possible only if the spacetime contains a global future and past (event) horizon, or if there is a disconnected patch of spacetime without a boundary, i.e. a closed universe. 

There is a different, intuitive way to understand why a bulk region causally disconnected from the boundary must exist if the global encoding is non-isometric. Recall that the  Euclidean wormhole contributions to the gravitational path integral that contribute to the inner product of states, thus signaling non-isometric encoding of the EFT, also lead to the existence of compact QESs. Since compact QESs must be causally disconnected from the boundary \cite{Engelhardt:2014gca}, wormhole corrections signaling non-isometry also imply that the bulk spacetime must contain a region causally disconnected from the boundary. 

With this intuition in mind, we can now formulate a conjecture that the presence of a global horizon implies that the global bulk-to-boundary map is non-isometric. Clearly, this conjecture requires some assumption about the bulk EFT Hilbert space in order to be true. As a trivial example of why this is the case, if we (artificially) define the bulk EFT Hilbert space of a black hole interior to not have any DOF whatsoever, the bulk-to-boundary map is (trivially) isometric. Similarly, if we consider a single-sided black hole and define the bulk EFT to only include ingoing modes which do not heavily backreact on the geometry, the whole spacetime remains in the domain of dependence of the causal wedge of the boundary for all states in the code subspace. Therefore, state-independent HKLL reconstruction can be successfully implemented and the global map is isometric. 

But what defines a ``physically reasonable'' bulk EFT Hilbert space? A naive choice could be to assign one degree of freedom to each Planck volume of a nice slice of the interior. This naturally leads to a strongly non-isometric map for a sufficiently late time slice, because the volume of such slices grows without bound. In the following conjecture we make a milder assumption, which nonetheless leads to a non-isometric map. We will focus here on the case of a connected bulk spacetime with an event horizon, and will comment on the closed universe case at the end of this subsection.

\begin{conjecture}
    Consider a black hole spacetime with a global (event) horizon. Then a) any bulk EFT Hilbert space should contain degrees of freedom associated with interior and exterior radially outgoing modes,\footnote{Of course, the radially outgoing modes inside the horizon never leave the black hole interior.} and b) any such bulk EFT Hilbert space is non-isometrically encoded in the fundamental black hole (i.e. boundary) Hilbert space.
    \label{conj:hor}
\end{conjecture}

Note first that the bulk EFT should be described, at leading order in $G$, by QFT in curved spacetime, for which outgoing modes in the interior and exterior are necessarily present. An infalling bulk observer should in principle be able to detect and alter these outgoing modes, which should therefore be treated as bulk EFT degrees of freedom. In particular, because the energy of a state with a fixed Rindler frequency is infinite at the horizon, these degrees of freedom should be identified with smeared wavepackets located a small but finite distance away from the horizon \cite{Engelhardt:2021qjs}. This is the rationale for  assumption $a)$ about the bulk EFT Hilbert space in our conjecture.

Let us now outline the intuitive reason for the conjecture $b)$. The authors of \cite{Engelhardt:2021qjs} showed that the inclusion in the bulk EFT Hilbert space, defined on a late enough time slice $\Sigma$, of degrees of freedom associated with outgoing interior wavepackets localized in a small interval (in Kruskal infalling coordinate $U$) near the horizon, is enough to guarantee, for some states in the bulk EFT Hilbert space, the existence of a non-trivial compact QES in the past of $\Sigma$. The QES sits at the black hole horizon up to corrections suppressed in $G$. This happens, surprisingly, even in one-sided black holes formed from collapse, in which compact QESs should naively not be present because the whole spacetime is causally connected to the boundary \cite{Engelhardt:2014gca,Engelhardt:2021qjs}. The reason is that the interior and exterior outgoing modes defining the code subspace are being treated independently, without requiring them to be entangled. As a consequence,  some states in the bulk EFT Hilbert space blueshift when evolved backwards in time, eventually backreacting non-perturbatively on the geometry and leading to the formation of a past horizon and a white hole singularity, therefore evading the argument for the non-existence of the QES \cite{Engelhardt:2021qjs}.\footnote{Note that the relative blueshifts of the ingoing and outgoing modes make it impossible to include them at the same time in a conventionally defined effective field theory below a reasonable energy cutoff \cite{Kiem:1995iy,Balasubramanian:1995sm}. This signals a breakdown of EFT near a black hole horizon, a phenomenon that has recurred in many guises since. The construction of \cite{Engelhardt:2021mue} partially avoids this issue by only considering wavepackets a finite distance away from the horizon. When disentangled, these wavepackets will not blueshift to transplanckian frequencies immediately, and the EFT is safe at least for an interval of time in the past of the initial slice. This is enough to argue for the existence of a compact QES \cite{Engelhardt:2021qjs}.}

Notice that this mechanism guarantees that the necessary condition for non-isometric encoding that we have argued for above---namely the existence of this past horizon, and therefore of a region causally disconnected from the boundary for some states in the bulk EFT Hilbert space---is satisfied for at least some states in the bulk EFT Hilbert space. This point is quite subtle. As usual, we are defining the bulk EFT Hilbert space on a (late) time slice, whose geometry is fixed for all states in the bulk EFT Hilbert space. We can think of this EFT, as we did before, as some matter DOF on top of a fixed geometry. However, to discuss causality we necessarily need to investigate global spacetime properties, which, due to backreaction of the matter on the geometry, depend on the particular state we specify on our time slice. As a result, all states in the bulk EFT Hilbert space will have the same geometry of the time slice we are considering, but different spacetime geometries in general. For example, if we start from a time slice of a black hole formed from collapse, some states in the bulk EFT will obviously have only a future horizon and no past horizon, with the entire spacetime causally connected to the boundary. But other states will contain a past horizon, a compact non-minimal QES, and a region causally disconnected from the boundary as advocated in \cite{Engelhardt:2021qjs}.

So a necessary condition for non-isometry is satisfied. Let us now give some intuition for why we expect our assumptions to actually be sufficient to guarantee non-isometric encoding. As we have discussed above, in all known examples, compact QESs arise thanks to  wormhole contributions to the gravitational path integral that are also responsible for correcting the inner product between bulk EFT states,  thus signaling non-isometry. Therefore, the results of \cite{Engelhardt:2021qjs} suggest that the inclusion of even a small number of DOF associated with outgoing interior modes leads to a non-isometric map. This in turn implies that the reconstruction of bulk EFT operators acting on at least some portion of the interior (namely, the part of the interior which is outside the simple wedge for all states in the bulk EFT Hilbert space) is state-dependent.

Notice that in the specific construction of \cite{Engelhardt:2021qjs}, the bulk-to-boundary map is weakly non-isometric, as we have already discussed. This is due to the fact that the outgoing wavepackets are confined to a small enough interval inside the horizon that their number is smaller than $S_{BH}$. In this case, the associated QES is non-minimal, and the bulk spacetime has a Python's lunch. By considering a larger number of outgoing wavepackets degrees of freedom, the map can  be made strongly non-isometric, and the QES minimal for appropriate states in the bulk EFT Hilbert space. For instance, this is the case for an evaporating black hole after the Page time, in which the number of outgoing interior DOF is identified with the number of interior Hawking partners, which exceeds $S_{BH}$ after the Page time. 

 %Suppose we have a holographic theory dual to some asymptotically AdS bulk spacetime plus a closed universe.  If there is no entanglement between the AdS spacetime and the closed universe, the fundamental Hilbert space of the closed universe is believed to be one-dimensional {\color{red} \cite{marolf-maxfield}}. Therefore, if any bulk EFT DOF in the closed universe is considered, the bulk-to-boundary map is strongly non-isometric. {\color{blue} If, on the other hand, there is bulk entanglement $S$ between the closed universe and the AdS spacetime 
Finally, consider a closed universe with  fields entangled with bulk fields in an asymptotically AdS spacetime. This bulk EFT is dual to a CFT defined on the boundary of the asymptotically AdS spacetime, see e.g., \cite{Antonini:2023hdh}.\footnote{A possible ambiguity in the holographic dictionary for this setup, which applies also to evaporating black holes, was discussed in \cite{antonini2024holographic}.} The kernel of the associated bulk-to-boundary map is empty if a small number $N<S$ of bulk EFT DOF in the closed universe is included \cite{Antonini:2023hdh}, where $S$ is the microcanonical entropy of bulk fields (or equivalently, of the CFT). However, in this case a non-minimal QES for the whole boundary of AdS exists, which has vanishing area term and ``excludes'' the closed universe (intuitively, it simply cuts the entanglement lines between the closed universe and the AdS spacetime). In the setup of \cite{Antonini:2023hdh}, such a QES also arises from a wormhole contribution to the gravitational path integral, with the map being weakly non-isometric. We expect this feature to arise generally in constructions involving closed universes.\footnote{An ongoing puzzle for the field \cite{antonini2024holographic} is to ask how the EFT in a closed universe---which can be defined using the gravitational path integral in the saddle-point approximation \cite{Antonini:2023hdh}---can arise if the underlying fundamental quantum gravity Hilbert space of the closed universe is one-dimensional, as suggested by non-perturbative gravitational path integral considerations \cite{marolf2020transcending}. A related interesting question is how the experience of a local bulk observer in the closed universe arises \cite{Chandrasekaran:2022cip,DeVuyst:2024pop}, and whether it can be described holographically \cite{antonini:WIP}.}

%\footnote{An ongoing puzzle for the field is to ask how the EFT in a closed universe can act as if it has one-dimensional Hilbert when studied alone, but nevertheless can act as if it has a non-zero entropy when interacting with an observer \cite{Chandrasekaran:2022cip} or when entangled with the EFT degrees of another universe \cite{Antonini:2023hdh,Balasubramanian:2023xyd,Balasubramanian:2020xqf}.} {\color{red} While I agree with everything said here, it is extremely odd that a closed universe would have a one-dimensional Hilbert space, but then somehow allow a large amount of entanglement with an AdS space, in which case it acts like it has more degrees of freedom. I'm sure how to make this more palatable, but I have added a footnote to discuss. Please take a look.} 

The necessity of state-dependent reconstruction for operators behind black hole horizons has been advocated for in the literature from several different, not always compatible, perspectives, see e.g. \cite{Almheiri:2013hfa,Marolf:2013dba,papadodimas2013infalling,papadodimas2014black,papadodimas2014state,Hayden:2018khn,Kourkoulou:2017zaj,Almheiri:2018ijj}. Given that in Sec.~\ref{sec:QI} we established an equivalence between state-dependence and non-isometry, our Conjecture \ref{conj:hor}, if proven to be true, leads to the same conclusion on general grounds.

%\sa{Comment about Papadodimas-Raju.}
%Our construction for the state-dependent interior operator has some similarities and key differences with the previous proposal of Papadodimas and Raju \cite{papadodimas2013infalling,papadodimas2014black,papadodimas2014state}\footnote{It is worth noting that their proposal has primarily engineered for non-evaporating black holes in AdS.}. The authors of \cite{akers2022black} have already, to fair extent, compared and contrasted their non-isometric state-dependent construction with the PR proposal.  Their discussion of the mirror operator ambiguity,  linearity of the state-dependent reconstruction and the uniqueness of the physical interpretation of the state in the fundamental description may be carried over to our framework. Therefore there is no need to recount them in here…..

To summarize, a non-isometric global bulk-to-boundary map can only arise in a setup in which a bulk subregion is causally disconnected from the boundary. Conversely, under reasonable assumptions for the definition of the bulk EFT Hilbert space on a sufficiently late time slice, we conjectured that the presence of a global horizon or a closed universe in the bulk spacetime implies a non-isometric global bulk-to-boundary map. 

\section{Discussion}
\label{sec:discussion}

In this paper we established a precise, quantitative relationship between non-isometric linear maps between Hilbert spaces and state-dependent unitary operator reconstruction. In the context of holography, we showed that weakly non-isometric bulk-to-boundary maps approximately preserving inner products between semiclassical states are approximately isometric, with the unitary reconstruction of bulk operators being approximately state-independent. On the other hand, maps with a non-trivial kernel (i.e., strongly non-isometric) unavoidably imply state-dependent reconstruction, even if the inner product between any two semiclassical states is approximately preserved under the map. Finally, we argued that non-isometry and state-dependence imply the presence of a bulk subregion causally disconnected from the boundary where the dual holographic theory is defined, and conjectured, under reasonable assumptions for the definition of the bulk EFT, that the presence of a global horizon guarantees a non-isometric global bulk-to-boundary map.

Several questions remain to be answered to fully understand the role of non-isometry and state-dependence in holography. First, as we have mentioned in Sec. \ref{sec:approximate}, there are setups in which the corrections computed by the gravitational path integral to overlaps between semiclassical states are not exponentially suppressed in $1/G$, but rather polynomial in $G$ or even order 1. Typical examples include black holes near the endpoint of evaporation, or closed universes entangled with asymptotically AdS spacetimes \cite{Antonini:2023hdh}. An even more extreme case is provided by closed universes unentangled with any asymptotically AdS spacetime. In that case, the Hilbert space of the closed universe sector is expected to be simply one-dimensional, and therefore the overlap between any two bulk EFT states is trivially 1 \cite{antonini2024holographic, marolf2020transcending}. 

Clearly, this seems to signal a breakdown of bulk EFT, which does not provide a good low-energy approximation of the full, UV complete fundamental theory, not even at leading order in $G$. However, this breakdown is somewhat surprising, because it arises well before the spacetime curvature becomes Planckian. This puzzle is particularly sharp for closed universes, where one could have a macroscopic closed universe with parametrically large curvature radius, and nonetheless a one-dimensional fundamental Hilbert space. In this situation, it should be possible to make sense of ordinary low-energy physics at least on a local scale. In fact, as far as we know, we could live in a closed universe. Within the framework of holography, however, it is not clear how this local physics would emerge from a fundamental description in terms of the dual holographic theory \cite{antonini2024holographic}. It remains to be understood whether and how holography can capture the physics of local observers in a closed universe, and, similarly, in the black hole interior at the late stages of evaporation.

Another important open question is whether a natural definition of the bulk EFT Hilbert space exists for a given holographic system, and whether for such a definition the presence of a global horizon or a closed universe implies a non-isometric bulk-to-boundary map. Motivated by QFT in curved spacetime, in Sec. \ref{sec:causal} we suggested that, in the presence of a horizon, radially outgoing modes in the interior should be part of the bulk EFT Hilbert space and conjectured (using the results of \cite{Engelhardt:2021qjs}) that the associated bulk-to-boundary map is necessarily non-isometric. A deeper understanding of this issue is needed to shed new light on the relationship between non-isometry and horizons.

%Let us formulate here a possible alternative (and in some sense opposite) definition of what a ``natural'' bulk EFT Hilbert space is.

%For concreteness, let us consider a bulk EFT given by QFT coupled to perturbative quantum gravity defined on top of a fixed semiclassical background with a global horizon. This is the case in most settings in which non-isometric maps are discussed. As we have seen in Section \ref{sec:causal}, the results of \cite{Engelhardt:2021qjs} show that the inclusion of outgoing interior modes as DOF in the bulk EFT Hilbert space unavoidably leads to non-perturbative backreaction of these DOF on the metric. For instance, for given states in the bulk EFT Hilbert space, a past singularity is present even though the bulk EFT Hilbert space was initially defined on a time slice of a single-sided black hole formed from collapse. This implies that our perturbative bulk EFT cannot self-consistently predict (backward) time evolution for such states, because the backreaction leading to the new geometry with a past singularity cannot be captured within perturbation theory. Is it then consistent to use this EFT past its breakdown, and conclude that a QES exists and a non-isometric map is present? On one hand, one could argue that in order for the EFT to be self-consistent possible criterion for the definition of the bulk EFT Hilbert space is to restrict the DOF such that there are no states leading to non-perturbative backreaction. 

%\sa{Maybe insert here code subspace compatible with time evolution, if we can consistently define that proposal?}

Finally, the results of Sec.~\ref{sec:causal} suggest a strict relationship between global horizons and global non-isometry. 
%Singularity theorems \cite{wald2010general,penrose1965gravitational,bousso2022singularities,wall2013generalized,bousso2023quantum} establish in turn the existence of a future singularity in the presence of a global horizon. \sa{Need to double check this.} 
Spacetimes with a global horizon typically contain singularities.\footnote{A noteworthy exception is given by the Bardeen black hole spacetime \cite{Bardeen,Borde:1994ai,Borde:1996df,Rodrigues:2018bdc} and similar regular black hole solutions \cite{Fan:2016hvf}, which contain a global event horizon and trapped surfaces behind them, but are nonetheless regular. However, the Bardeen black hole spacetime is not globally hyperbolic, and therefore is able to escape singularity theorems despite the presence of trapped surfaces \cite{Borde:1994ai,Borde:1996df}. We expect that the presence of a global horizon implies the existence of trapped regions behind it, which in turn implies the existence of a singularity if the spacetime is globally hyperbolic \cite{wald2010general,penrose1965gravitational,bousso2022singularities,wall2013generalized,bousso2023quantum}. We leave a more thorough investigation of this point to future work. We thank Raphael Bousso, Geoff Penington, and Arvin Shahbazi-Moghaddam for a discussion on this point.} Note that closed universes with a negative cosmological constant, which are those arising in the setups discussed in  \cite{Antonini:2023hdh,marolf2020transcending}, are Big Bang-Big Crunch universes, and therefore also contain singularities. This suggests a connection between singularities and non-isometry, and possibly that the puzzles arising in the presence of non-isometry and state-dependence---such as the large overcounting of physical states within EFT in the strongly non-isometric case, or the breakdown of bulk EFT discussed above---have their origin in our lack of understanding of the physics of the singularity. 
%\sa{I'm not sure what we're trying to say with this. I think it gets a bit away from the point and confusing. I would personally erase this. Indeed, as we expect from the chaotic aspects of quantum gravity, the underlying time evolution is approximately pseudorandom on an appropriate code subspace for some theories. There the increasing complexity of time-evolving states seems to result in the effective presence of a horizon and a singularity \cite{Engelhardt:2024hpe}.  Following our arguments, the code subspace would then also have a non-isometric embedding and require state-dependent reconstruction.}
The final state proposal \cite{horowitz2004black} is an example of a mechanism through which physics near the singularity could drastically reduce the number of physical DOF in the EFT, and possibly resolve some of these puzzles. The price to pay is  the definition of an unconventional bulk EFT. It would be interesting to understand better what precise role the singularity plays in holographic systems with a non-isometric map.

\section*{Acknowledgements}
We would like to thank Shadi Ali Ahmad, Raphael Bousso, Zuzana Gavorova, Luca Iliesiu, Guanda Lin, Javier Magan, Geoff Penington, Pratik Rath, Suvrat Raju,  Martin Sasieta,  Arvin Shabhazi-Moghaddam, and Thomas Vidick for useful discussions. S.A. is supported by the U.S.
Department of Energy through DE-FOA-0002563. V.B. was supported in part by the DOE through DE-SC0013528 and QuantISED grant DE-SC0020360, and in part by the Eastman Professorship at Balliol College, Oxford. C.C. acknowledges support by the Commonwealth Cyber Initiative at Virginia Tech.
N.B. is supported by the U.S. Department of Energy through the Quantum Telescope Project.

\appendix
\section{Useful lemmas and definitions}
These are used in the proofs in Appendix~\ref{app:qiproof}. We keep them for completeness and convenience.
\begin{lemma}[von Neumann's Trace inequality]
    Consider two complex $n\times n$ matrices $A,B$ whose singular values are $\alpha_1\geq \alpha_2\geq \dots\geq \alpha_n$ and $ \beta_1\geq \beta_2\geq \dots\geq \beta_n$ respectively. Then 
    $$|\Tr[AB]|\leq \sum_i \alpha_i\beta_i$$.
\end{lemma}

\begin{definition}[Nuclear norm]
    The nuclear norm of a matrix $A$ is 
    $$||A||_* = \Tr[\sqrt{A^{\dagger}A}]=\sum_{i}|\alpha_i|$$
    where $\alpha_i$ are the singular values of $A$. 
\end{definition}
\begin{definition}
   The entrywise 1-norm of a matrix $A$ is $$||A||_{E_1}=\sum_{ij}|a_{ij}|.$$
\end{definition}
\begin{lemma}[Bound on nuclear norm]
    Let $A$ be an $n\times m$ complex matrix, $$\frac{1}{\sqrt{mn}}||A||_{E_1}\leq ||A||_F\leq||A||_*\leq ||A||_{E_1}$$
    %\label{lemma:singupper}
    \label{lemma:schatten1}
\end{lemma}
\begin{proof}
    Let $\vec{\alpha}=(\alpha_1, \alpha_2,\dots, \alpha_n)$ be the singular values of $A$. We know that $||A||_F=||\vec{\alpha}||_2=\sqrt{\sum_i|\alpha_i|^2}\leq ||\vec{\alpha}||_1=||A||_*$ and $||A||_{E_1}/\sqrt{mn}\leq ||A||_F$ follows from Cauchy-Schwartz inequality. Here we used $||\cdot||_p$ to denote the vector $p$-norm, not to be confused with the matrix $p$-norm in the main text. Note that the entrywise 1-norm upper bound is saturated when the only non-zero matrix element is $A_{11}=1$. The lower bound is saturated when all entries are 1.   %\CC{rmb to rephrase or rewrite. Also how do you cite this answer} %The original question asks for a real matrix, but I don't see any indication in the proof why it shouldn't generalize to a complex matrix.

%The best such inequality that depends only on $m$ and $n$ is the one above. The right inequality is tight when $A$ is a matrix with a $1$ in the top-left corner and zeroes elsewhere. The left inequality is tight when $A$ is the matrix all of whose entries are 1. These examples also show that you cannot get any better constants even if you let them depend on the rank of $A$.

%To see that the inequalities actually hold, let's start with the left inequality. Write $A$ in its singular value decomposition:
%$$A=\sum_i \sigma_i \mathbf{x}_i\mathbf{y}_i^*,$$
%where $\{\sigma_i\}$ are the singular values of $A$ and $\{\mathbf{x}_i\}$ and $\{\mathbf{y}_i\}$ all have Euclidean norm 1: $||\mathbf{x}_i||=||\mathbf{y}_i||=1$ for all $i$. Then $$||A||_1=||\sum_i\sigma_i \mathbf{x}_i\mathbf{y}_i||_1\leq \sum_i\sigma_i||\mathbf{x}_i\mathbf{y}_i||_1\leq \sqrt{mn}\sum_i\sigma_i||\mathbf{x}_i\mathbf{y}_i||_2=\sqrt{mn}\sum_i\sigma_i=\sqrt{mn}||A||_*$$ where $||\cdot||_2$ is the entrywise 2-norm or the Frobenius norm. 

For the upper bound $||A||_{E_1}\geq ||A||_*$, a proof can be found \href{https://mathoverflow.net/questions/201845/equivalence-of-entrywise-1-norm-and-schatten-1-norm}{here}, but we reproduce it below for convenience. It is shown by \cite{johnston2013duality} (see pages 3 and 4) that the Schatten 1-norm is an infimum over all rank-1 decompositions of A 
$$||A||_*=\inf\{\sum_i|c_i|: A=\sum_{i}c_i\mathbf{x}_i\mathbf{y}_i^*,||\mathbf{x}_i||=||\mathbf{y}_i||=1\},$$
where all vectors $\mathbf{x}_i,\mathbf{y}_i$ have unit Euclidean norm.
One such rank-1 decomposition of $A$ is just the very naive one that writes it in terms of standard basis vectors $\{\mathbf{e}_i\}$:
$$A=\sum_{ij}A_{ij}\mathbf{e}_i\mathbf{e}_j^*.$$
Since $||A||_{E_1}=\sum_{ij}|A_{ij}|$, it follows immediately that $||A||_*\leq ||A||_{E_1}$.

\end{proof}

\begin{lemma}
    Let $B$ be any $n\times n$ complex matrix and $\tilde{O}$ is unitary,
     then 
     \begin{equation}
         \max_{\tilde{O}\in U}\Re\{\Tr[\tilde{O}B]\} = ||B||_*
     \end{equation}
     \label{lemma:spectralNorm}
\end{lemma}
\begin{proof}
Since $\tilde{O}$ is unitary, its singular values are 1s. Let $\beta_i$ be the singular values of $B$, then applying von Neumann's Trace inequality, we obtain an upper bound $$\max_{\tilde{O}}Re\{\Tr[\tilde{O}B]\}\leq \max_{\tilde{O}}|\Tr[\tilde{O}B]|\leq \sum_{i}\beta_i = ||B||_*.$$ 
%    We have previously established that $||B||_*$ is an upper bound by von Neumann's trace inequality.
    Since $B=W\Sigma V^{\dagger}$ is a square matrix and admits a singular value decomposition, let $\tilde{O}_0=VW^{\dagger}\in U(n)$ such that $\Tr[\tilde{O}_0B]=\Tr[\Sigma]=||B||_*$. Since we chose a particular $\tilde{O}$, this must be a lower bound on the maximum.
\end{proof}

%\begin{lemma}
%Let $A, B$ be $n\times n$ matrices, 
%    and $a_{ij}$ denote matrix elements, then $$\sum_{i}|a_{ii}|\leq ||A||_*$$ and $$\Tr[AB]\leq ||A||_*||B||_{\infty}$$
%\label{lemma:abstracebound}
%\end{lemma}
%\begin{proof}
%A short sketch can be found \href{https://math.stackexchange.com/questions/2451958/lower-bound-nuclear-norm-of-a-by-mathrmtra}{here}.
%\end{proof}

\begin{proposition}
    Let $\tilde{O}$ be an $n\times n$ unitary and $R$ be an $n\times m$ matrix with $m\geq n$. Let $B(k)$ be an $n\times n$ submatrix of $R$ by keeping the $i+k$th to $(i+k+n-1)mod(m) +1$th column of $R$. Then $$\sum_{k=1}^m \max_{\tilde{O}}Re\{\sum_{i=1}^n[\tilde{O}R]_{i,i+k}\}= \sum_{k=1}^m||B(k)||_*\leq n||R||_{E_1}.$$
%    When the index value $>m$, we take its value mod $m$. 

  %  Update note: The last inequality is further improved to $\leq n||R||_1$ using Lemma~\ref{lemma:schatten1}.
\end{proposition}

\begin{proof}
Note that each $B(k)$ is just taking a $n\times n$ submatrix of $R$ and each $k$ marks a cyclic permutation around the row index of $R$. The left equality follows from Lemma~\ref{lemma:spectralNorm}.
Now we use Lemma~\ref{lemma:schatten1} so that $||B(k)||_*\leq ||B(k)||_{E_1}$.
    \begin{align}
        &\sum_{k=1}^{m}||[B(k)]||_*\leq \sum_{k=1}^m ||B(k)||_{E_1} = \sum_k \sum_{ij}|B(k)_{ij}|\\
        &= \sum_{ij}(\sum_k|B(k)|_{ij})= n \sum_{ij}|R_{ij}| = n||R||_{E_1}
    \end{align}
    where for the first term on the second line, we summed over $k$ first and recognized that this is simply summing over the absolute values of all matrix elements of $R$ $n$ times as each element is contained in $n$ such $B(k)$ matrices.

  %  Note that the first inequality uses lemma~\ref{lemma:singupper}, but using lemma~\ref{lemma:schatten1} shaves off the $\sqrt{n}$ factor, leading to an improved result $n||R||_1$ at the very end.
\end{proof}

\begin{lemma}
    If $R$ is an $n\times m$ matrix with $m\geq n$ such that the rows of $R$ are orthonormal (where vector norms are given by the Euclidean norm $||\cdot||_2$), then $||R||_{E_1}\leq n\sqrt{m}$.
\end{lemma}
\begin{proof}
Recall that $||v||_1\leq \sqrt{d}||v||_2$ where the vector $v$ has dimension $d$. Then simply 
    $||R||_{E_1}=\sum_{ij}|R_{ij}|= \sum_i \sum_j|R_{ij}| \leq \sum_i \sqrt{m}=n\sqrt{m}$.
\end{proof}

\begin{corollary}\label{cor:coisobound}
For an $n\times m$ matrix $R$ whose rows are orthonormal and $n\leq m$, we then know that $\sum_{k=1}^m\max_{\tilde{O}}Re\{\sum_{i=1}^n[\tilde{O}R]_{i,i+k}\}\leq n^{2}\sqrt{m}$.
\end{corollary}

The bound is observed to obey a scaling of $n^{3/2}\sqrt{m}$ in small scale numerics, suggesting that Corrollary \ref{cor:coisobound} may not be tight.
%\sa{I think this was a note-to-self maybe? We don't refer to numerics before. Maybe replace it with something else/erase it?: This is still not quite as tight as the bound we observed numerically which seems to scale as $n^{3/2}\sqrt{m}$. However, the power of $n$ is not very sharp from the numerics. The power of $m$ is very much $1/2$. Using this bound, we can derive a bound for the state-dependence measure when $V$ is co-isometric.}

Now consider a slightly more general version 
\begin{lemma}\label{lemma:upperbound}
    Let $\tilde{O}$ be an $n\times n$ unitary and $R$ be an $n\times m$ matrix with $m\geq n$. Let $B(k)$ be an $n\times n$ submatrix of $R$ by keeping the $i+k$th to $(i+k+n-1)mod(m)+1$th column of $R$, and $\tilde{D}$ be a diagonal matrix with diagonals $\sigma_1,\sigma_2,\dots,\sigma_{n}$. Then $$\sum_{k=1}^m \max_{\tilde{O}}Re\{\sum_{i=1}^n[\tilde D\tilde{O}\tilde DR]_{i,i+k}\}\leq \sum_{k=1}^m||\tilde DB(k)\tilde D||_*%\leq \sqrt{mn}(\sum_{k=1}^n\sigma_k)^2
    \leq \sqrt{m}(\sum_{k=1}^n\sigma_k)^2$$
    %When the index value $>m$, we take its value mod $m$. 
\end{lemma}

\begin{proof}
Note that the expression inside the curly bracket can be written as $\Tr[\tilde D\tilde{O}\tilde DB(k)]=\Tr[\tilde{O}\tilde DB(k)\tilde D]$. Let $B_D(k)=\tilde D B(k) \tilde D$.
    Again, from the von Neumann trace inequality.
    $$\max_{\tilde{O}}Re\{\Tr[D\tilde{O}DB(k)]\}\leq \max_{\tilde{O}}|\Tr[\tilde{O}B_D(k)]|\leq \sum_{i}\beta_i(k) = ||B_D(k)||_*.$$ Hence the left inequality follows for each $k$ and trivially so for the sum.

Now we use the lemma  \ref{lemma:schatten1} above so that $||B_D(k)||_*\leq ||B_D(k)||_{E_1}$.
    \begin{align}
        &\sum_{k=1}^{m}||[B_D(k)]||_*\leq \sum_{k=1}^m  \sum_{k}||B_D(k)||_{E_1} = \sum_k \sum_{ij}|B_D(k)_{ij}|\\
        &= \sum_{ij}(\sum_k|B_D(k)|_{ij})=  \sum_{ij}|R_{ij}|\sigma_i(\sum_k\sigma_k)
    \end{align}
For the final sum, we see that for an element in the same row, it appears in the sum $|B|$ times each time with a different weight $\sigma_k$.

Since each row vector of $\vec{u}_i=R_{ij}$ has unit 2-norm, we have $$\sum_{ij}|R_{ij}|\sigma_i= \sum_i ||u_i||_1\sigma_i\leq \sqrt{m}\sum_{i}\sigma_i.$$
Putting everything together and applying Lemma~\ref{lemma:schatten1}, the right inequality reads 
$$\sum_k||B_D(k)||_*\leq (\sum_{i}\sigma_i||u_i||_1)(\sum_k\sigma_k)\leq \sqrt{m}(\sum_{k=1}^n\sigma_k)^2.$$

%Using the improved bound in lemma~\ref{lemma:schatten1}, we have $||B_D(k)||_*\leq ||B_D(k)||_1$ and thereby removing the $\sqrt{n}$ factor at the very end.

%$$\sum_k||B_D(k)||_*\leq (\sum_{i}\sigma_i||u_i||_1)(\sum_k\sigma_k)\leq \sqrt{m}(\sum_{k=1}^n\sigma_k)^2.$$

\noindent This recovers the bound from Corollary~\ref{cor:coisobound} as a special case when $\sigma_k=1$.

\end{proof}

Now we lower bound $\sum_{k=1}^m \max_{\tilde{O}}Re\{\sum_{i=1}^n[\tilde{D}\tilde{O}\tilde{D}R]_{i,i+k}\}$. %Note that \href{https://yingzhouli.com/posts/2020-07/von-neumann-trace-inequalities.html#ruhe} {Ruhe's trace inequality} does not directly apply as $B$ is not hermitian nor positive semi-definite\cite{marshall11}. However, a simpler argument exists since we are doing maximization.

\begin{theorem}
Let $\tilde{O}$ be an $n\times n$ unitary and $R$ be an $n\times m$ matrix with $m\geq n$. Let $\tilde{D}=diag(\sigma_1,\sigma_2,\dots)$ be a diagonal matrix, $B(k)$ be an $n\times n$ submatrix of $R$ as before. Then
$$\frac{1}{n}(\sum_k\sigma_k)^2\leq\sum_{k=1}^m||\tilde{D}B(k)\tilde{D}||_*=\sum_{k=1}^m\max_{\tilde{O}}\Re\{\sum_{i=1}^n[\tilde{D}\tilde{O}\tilde{D}R]_{i,i+k}\}$$   %\CC{in progress}
\end{theorem}

\begin{proof}
Recall that $B_D(k)=\tilde{D}B(k)\tilde{D}$ and $||B_D(k)||_*\geq ||B_D(k)||_F\geq \frac{||B_D(k)||_{E_1}}{n}$ from Lemma~\ref{lemma:schatten1}.

Setting $u_i$ to be the row vectors of $R_{ij}$ which have unit 2-norms,
\begin{align}&\frac{1}{n}\sum_{k}||B_D(k)||_{E_1} = \frac{1}{n}\sum_{ij}|R_{ij}|\sigma_i\sum_k\sigma_k = \frac{\sum_k\sigma_k}{n}\sum_i \sigma_i||u_i||_1\\
&\geq \frac{\sum_k\sigma_k}{n}\sum_i\sigma_i||u_i||_2=\frac{(\sum_k\sigma_k)^2}{n}\end{align}
\end{proof}

%Again I do not think this bound is tight. 
%This bound should have room for improvement. Case in point is when all $\sigma_k=1$. Empirically when the entries of $R$ is either 0 or 1, the quantity attains $n^2$, not the above lower bound $n$. %Intuitively such configurations give a lower value compared to when all entries of $R$ are $\pm 1$ up to normalization, which should yield the upper limit. Although I struggle to find counterexamples that violate the $n^2\leq $ bound, I could not prove this tighter bound. I suppose the key to proving better bounds involve finding the singular values of $B(k)$ without reference to the Schatten 1 norm.

%From Lemma ~\ref{}, it is obvious by cyclic property of trace that 

\section{Proofs of QI Statements}\label{app:qiproof}
We will order the proofs below based on dependencies, so they may not appear in the same order as they do in the main text.

\subsection{Proof of Theorem \ref{thm:d1X0}}
\label{sec:thmd1X0}

\begin{theorem}
    Let $V:A\rightarrow B$ be a linear map  and let $\chi_0$ be a unitary 1-design over $A$  such that $\chi_0=\{QPQ^{\dagger}: P\in \mathcal{P}, \, Q^{\dagger}|\psi_i\rangle=|i\rangle\}$ where $\mathcal{P}$ is the generalized Pauli group over $A$, $\{|\psi_i\rangle\}$ is the eigenbasis of $V^{\dagger}V$, and $\{|i\rangle\}$ is the computational basis spanning $A$. Then
 \begin{equation}\label{Defd1app}
        d_1(V,\chi_0) \equiv \frac{1}{|\chi_0|}\sum_{O\in \chi_0} \Big[ \min_{\tilde{O}\in U(B)} D_V(\tilde{O},O)\Big] = \frac{1}{|A|}\sum_{i,j=1}^{|A|}(\sigma_i-\sigma_j)^2
    \end{equation}
    where $\{\sigma_i\}$ are the singular values of $V$ and $|\chi_0|$ is the number of elements in the set $\chi_0$.  For the basis vectors that span the null space we take the singular value to be $0$.
\end{theorem}
\begin{proof}
We give a detailed proof below.    Recall that for $T=\tilde{O}V-VO$,
    \begin{align}\label{DvOO}
    D_V(\tilde{O},O) &= \int d\mu(|\psi\rangle) ||T|\psi\rangle||^2\\
    &=\int d\mu(|\psi\rangle)\Tr[T^{\dagger}T|\psi\rangle\langle\psi|]\\
    &=\frac{1}{|A|}\Tr[T^{\dagger}T]\\
    &=\frac{1}{|A|} \sum_{i=1}^{|A|}\langle\psi_i|T^{\dagger}T|\psi_i\rangle \\
    &=\frac{1}{|A|}\sum_{i=1}^{|A|}||T|\psi_i\rangle ||^2
\end{align}
in the vector 2-norm $||\cdot||$, where we have chosen a basis $\{|\psi_i\rangle\}$ such that $V|\psi_i\rangle=\sigma_i|\tilde{\psi}_i\rangle\in B$ such that $\{|\tilde{\psi}_i\rangle\}, \{|\psi_i\rangle\}$ are the Schmidt bases for the Choi state of $V$. All states $\ket{\psi_i}$ and $|\tilde{\psi}_i\rangle$ are normalized to 1. If a state $|\psi_i\rangle$ is in the kernel of $V$, we take its corresponding singular value $\sigma_i$ to be $0$.

Then \begin{equation}
    D_V(\tilde{O},O) = \frac{1}{|A|}\sum_{i=1}^{|A|}||\sigma_i\tilde{O}|\tilde{\psi}_i\rangle - VO|\psi_i\rangle||^2
\end{equation}
Now consider the sum over $O\in \chi_0$. Because $\chi_0$ is a 1-design, we can choose $\chi_0=\{QPQ^{\dagger}: P\in \mathcal{P}\}$ where $\mathcal{P}=\{\omega(a,b)Z^bX^a, a,b=0,\dots, |A|-1, \omega(a,b)\in\mathbb{C}\}$ is the generalized Pauli group  that is orthonormal under the trace inner product $\Tr[P_1^{\dagger}P_2]/|A|=1$. Here $X,Z$ are shift and clock matrices such that in the computational basis $\{|i\rangle\}$, $X|i\rangle =|i\oplus 1\rangle$ shifts the computational basis and $Z|j\rangle = \exp(i\phi_j)|j\rangle$ adds a phase where $\oplus$ is addition modulo $|A|$. $Q$ is a unitary such that it maps between the Schmidt basis and the computational basis , i.e., $\forall i, Q^{\dagger}|\psi_i\rangle = |i\rangle$.

First consider the set of unitaries $O=QX^aQ^{\dagger}$, 

\begin{equation}
D_V(\tilde{O},O) = \frac{1}{|A|} \sum_i||\sigma_i \tilde{O}|\tilde{\psi_i}\rangle - \sigma_{i\oplus a}|\tilde{\psi}_{i\oplus a}\rangle||^2   
\label{eqn:sqrdiff}
\end{equation}
Formally, we see that the sum is minimized when each term in the summand
\begin{equation}
    \sigma_i^2+\sigma_{i\oplus a}^2-2Re\{\sigma_i\sigma_{i\oplus a}\langle\tilde{\psi}_i| \tilde{O}^{\dagger}|\tilde{\psi}_{i\oplus a}\rangle\}
\end{equation} is minimized, i.e., 
$\tilde{O}^{\dagger}|\tilde{\psi}_{i\oplus a}\rangle=|\tilde{\psi}_i\rangle$. However, we need to be a bit more careful here as some of the states are null when $|A|>|B|$. In the case where the kernel is trivial, whatever operator we implement in $A$ can be reconstructed in $B$ by choosing a unitary that permutes the target states in the right way. The same is not true for the case where there's a non-trivial kernel---$O$ can permute a state in $ker^{\perp}(V)$ to a state in $ker(V)$ and vice versa. In this case, no matter what $\tilde{O}$ we choose, it cannot possibly complete all the permutations as it only acts on the image of a subspace $B\cong ker^{\perp}(V)$. However, we will see that the end result is the same by choosing $\tilde{O}$ which correctly permutes the elements that stay within the image of $ker^{\perp}(V)$. Let $S, S_{\perp}$ denote the set of indices such that for all $i\in S_{\perp}$, $|\psi_{i\in S_{\perp}}\rangle\in ker^{\perp}(V)$ and for all $j\in S, |\psi_{j\in S}\rangle \in ker(V)$.

Consider the permutation $p$ of the indices by the action of $O$. Its action on $|\psi_i\rangle$ can be classified into 3 different types: (1) $i\in S_{\perp}$ and $p(i)\in S_{\perp}$. These vectors that are mapped from $ker^{\perp}(V)$ to itself span a subspace $A_r$. Let the action of $O$ restricted to this subspace $A_r$ be $O_{r}$, which can always be written as a unitary on this subspace as it just involves a sequence of transpositions of the indices. Each transition between two given indices $i\leftrightarrow j$ is represented by a unitary with entries $U_{ij}=U_{ji}=1$, $U_{kl}=\delta_{kl}$ for $k,l\ne i,j$, and the rest of the entries vanishing. (2) The states in the null space are mapped to its orthogonal complement and vice versa. (3) The set of states are mapped from $ker(V)$ to itself.

Now we claim that as long as $\tilde{O}$ permutes the vectors in $ker^{\perp}(V)$ the same way as $O_{r}$, then it minimizes the square difference in (\ref{eqn:sqrdiff}). Denote such an operator as $\tilde{O}_{r}$. For states $|\psi\rangle\in A_r$, $\tilde{O}_r$ clearly minimizes the square difference: 
\begin{equation}
    ||\tilde{O}_r\sigma_i|\tilde{\psi}_i\rangle - \sigma_{p(i)}|\tilde{\psi}_{p(i)}\rangle||^2 = ||\sigma_i|\tilde{\psi}_{p(i)}\rangle - \sigma_{p(i)}|\tilde{\psi}_{p(i)}\rangle||^2 =  (\sigma_i- \sigma_{p(i)})^2 
\end{equation}
%The unitary reconstruction of restricted permutation $O_r$ always exists because one can simply consider a restricted bijective map $f$ that maps $A_r$ to $f(A_r)\subset B$. And the reconstruction $\tilde{O}$ is then isomorphic to $O_r$ which is the unitary action defined on the subspace $A_r$.
%
For states that are in (3), no matter what we choose for $\tilde{O}$, $\sigma_i,\sigma_{p(i)}=0$. Hence their contribution to the sum is trivial: $||\tilde{O}_r\sigma_i|\tilde{\psi}_i\rangle - \sigma_{p(i)}|\tilde{\psi}_{p(i)}\rangle||^2=0$.

Finally for the states in (2): suppose $|\psi_i\rangle\in ker(V)$ and $|\psi_{\sigma(i)}\rangle\in ker^\perp (V)$. Then 
\begin{equation}
||\tilde{O}\sigma_i|\tilde{\psi}_i\rangle - \sigma_{p(i)}|\tilde{\psi}_{p(i)}\rangle||^2 = ||\sigma_{p(i)}|\tilde{\psi}_{p(i)}\rangle||^2 =  \sigma_{p(i)}^2
\end{equation}
no matter what $\tilde{O}$ we choose because $\sigma_i=0$. Similarly, if $|\psi_i\rangle\in ker^\perp (V)$ and $|\psi_{p(i)}\rangle \in ker(V)$, then the second term in the above equation vanishes. Then for any unitary $\tilde{O}$
\begin{equation}
||\tilde{O}\sigma_i|\tilde{\psi}_i\rangle - \sigma_{p(i)}|\tilde{\psi}_{p(i)}\rangle||^2 = ||\sigma_{i}\tilde{O}|\tilde{\psi}_{i}\rangle||^2 = \sigma_{i}^2.
\end{equation}
Therefore, summing over $O=QX^aQ^{\dagger}, a=0,1,\dots |A|-1$,
\begin{equation}
    \Sigma_X:=\sum_{O=QX^a Q^{\dagger}}\min_{\tilde{O}}D_V(\tilde{O},O)= \sum_{a}\sum_i (\sigma_i-\sigma_{i\oplus a})^2
\end{equation}
Recall that $p$ is generated by $X$, therefore $p_a(i)=i\oplus a$. Performing the $a$ sum first, we must have 
\begin{equation}
    \Sigma_X = \sum_i \sum_j(\sigma_i- \sigma_{i\oplus j})^2 = \sum_{i,k}(\sigma_i-\sigma_k)^2 = \sum_{i\ne j} (\sigma_i-\sigma_j)^2
\end{equation}
where $i,j\in \mathbb{Z}_{|A|}$. Note that $\sigma_i=0$ are included in the sum for basis states that span the kernel. For the second equality, we relabel the indices of $i\oplus j=k$ because summing over $j$ clearly yields all index values $<|A|$. Then we see that $\Sigma_X$ is simply computing the differences between all pairs of singular values of $V$ (including self-pairings, but since they always produce 0, we can safely remove them and have the third equality).

Let us now consider operators $O=QZ^bQ^{\dagger}$, which are element of the generalized Pauli group with $a=0, b\ne 0$. The action of $O=QZ^bQ^{\dagger}$ adds a pure phase to $|\psi_i\rangle$ but leaves the state invariant. 
The term of interest is therefore 
\begin{equation}
||\tilde{O}V|\psi_i\rangle-VO|\psi_i\rangle||^2=    \sigma_i^2||\tilde{O}|\tilde{\psi}_i\rangle - \exp(i\phi_i)|\tilde{\psi}_i\rangle||^2
\end{equation}
For the case of trivial kernel, $\tilde{O}$ are easy to reconstruct as we just need to choose $\tilde{O}=\tilde{Q}\tilde{Z}^b\tilde{Q}^{\dagger}$ that adds the same phase to $|\tilde{\psi}_i\rangle$, where $\tilde{Q}$ are unitaries that map the Schmidt basis in $B$ to the computational basis in $B$, $\tilde{Q}|\tilde{\psi}_i\rangle=|\tilde{i}\rangle$ and $\tilde{Z}^b$ are Pauli operators on $B$.  Then for each $i$ the square difference vanishes and the entire sum $\sum_{i}||\tilde{O}V|\psi_i\rangle-VO|\psi_i\rangle||^2=0$.

For any $V$ with kernel, we can discuss two cases: either $|\psi_i\rangle$ is in the kernel or it is not. %For the case it is not, it can be reduced to an argument similar to what we had before where $\tilde{O}$ just needs to reproduce the phases for a restricted operator $O_r$ that acts on $A_r=ker^{\perp}(V)$. 
For the states $|\psi_i\rangle$ that are in the kernel, the $O$ adds a phase to the state which is then mapped to 0. The corresponding $\tilde{O}V|\psi_i\rangle=0$ also for any $\tilde{O}$, hence these terms do not contribute. For $|\psi_i\rangle$ not in the kernel, we just choose $\tilde{O}_r$ such that it adds the same phase as $O$ acting on states in $ker^{\perp}(V)$, i.e.,  $\tilde{O}|\tilde{\psi}_j\rangle = \exp(i\phi_j)|\tilde{\psi}_j\rangle$. As such,  $$\Sigma_Z:=\sum_{b=0}^{|A|-1}\sigma_i^2 ||\tilde{O}|\tilde{\psi}_i\rangle -\exp(i\phi_i)|\tilde{\psi}_i\rangle||^2=0.$$ 
They do not contribute as all $O=QZ^bQ^{\dagger}$ are reconstructable operators in a state-independent manner. 

From the above, we see that for arbitrary $a,b$, we can think of the action of elements of the generalized Pauli group on $|\psi_i\rangle$ as permuting the eigenstates followed by adding a phase. As we can just construct $\tilde{O}$ that undoes the phase change and the permutations separately, as we have shown, it is clear that the reconstruction which minimizes the norm for each $i$ is the one that undoes the permutation as well as the phase change. Because all non-trivial phase operators can always be perfectly reconstructed, the contribution of $O=\omega(a,b)QZ^bX^aQ^{\dagger}$ to the measure is identical to that of $QX^aQ^{\dagger}$. Therefore, the sum over powers $b=0,...,|A|-1$ of $Z$ for operators of the form $O=\omega(a,b)QZ^bX^aQ^{\dagger}$ in the measure \eqref{Defd1app} will simply result in a multiplicative factor $|A|$. Note that the global phase $\omega(a,b)$ will not impact our argument in any way.  Therefore, we finally obtain
\begin{equation}
    d_1(V,\chi_0)=\frac{1}{|A|^2}\sum_{O\in \chi_0}[\min_{\tilde{O}}D_V(\tilde{O},O)] = \frac{1}{|A|^2} |A| \Sigma_X = \frac{1}{|A|}\sum_{i,j}(\sigma_i-\sigma_j)^2
\end{equation}
where we used $|\chi_0|=|A|^2$.

\begin{comment}
    
Now we show that the above value is indeed the minimum over all 1-designs $X$. Without loss of generality, consider any 1-design $X'$ which can be obtained from $\chi_0$ by a left unitary multiplication so that for any $X'=\{O': O'=UO\}$. 

From (\ref{DvOO}), we see that 
\begin{equation}
    D_V(\tilde{O},O') = 2\Tr[V^{\dagger}V] - 2Re\{\Tr[V^{\dagger}\tilde{O}VO']\}
\end{equation}

Minimizing $D_V$ means we maximize \begin{equation}
    Re\{\Tr[V^{\dagger}\tilde{O}VO']\} = Re\{\sum_i\sigma_i\langle\tilde{\psi}_i|\tilde{O}VUO|\psi_i\rangle\}
\end{equation}

Again, let $O=QX^aZ^bQ^{\dagger} = \omega(a,b) QZ^b X^a Q^{\dagger}$ where $\omega(a,b)$ is some complex phase.

The quantity inside the curly bracket now reads
\begin{align}
&\omega(a,b)\sum_i \sigma_i \langle \tilde{\psi}_i| \tilde{O} VUQZ^a |\psi_{\sigma(i)}\rangle  \\
=&\omega(a,b)\sum_i\sigma_i\langle\tilde{\psi}_i|\tilde{O} V\sum_{j} U'_{j\sigma(i)}|\psi_j\rangle\\
=&\omega(a,b)\sum_{i,j} \sigma_i\sigma_{j} \langle\tilde{\psi}_i|\tilde{O}|\tilde{\psi}_j\rangle U'_{j\sigma(i)}\\
=&\omega(a,b)\sum_{ij}\sigma_i\sigma_j \tilde{O}_{ij}U'_{j\sigma(i)}
\end{align}
The permutation $\sigma$ depends on $b$ and we have rewritten $UQZ^a = U'(a)$. Note that $\tilde{O}_{ij}$ is a unitary and the element vanishes for any $i,j$ that labels null states, therefore $M_{i\sigma(i)}=\sum_j\tilde{O}_{ij}U'_{j\sigma(i)}$ must have $|M_{kl}|\leq 1$ and the norm of each row and column vector $\leq 1$. Furthermore, any row $i$ with $\sigma_i=0$ must be 0. \color{blue}{in progress}
\end{comment}
\end{proof}

\begin{remark}
    Note that some but not all observables are reconstructable even in a state-specific manner. From the above proof, we can identify which ones they are in the 1-design. For example, with respect to any state $|\psi\rangle \in ker^{\perp}(V)$, the $O_Z=QZ^bQ^{\dagger}$ is reconstructable as $\tilde{O}_Z=\tilde{Q}\tilde{Z}^b\tilde{Q}^{\dagger}$ because we see that $VO_Z|\psi_i\rangle= \tilde{O}_ZV|\psi_i\rangle$ for all $|\psi_i\rangle\in ker^{\perp}(V)$. However, for any state $|\psi_i\rangle\in ker(V)$ such that $QX^aQ^{\dagger}|\psi_i\rangle \in ker^{\perp}(V)$ (or for any $|\psi_i\rangle\in ker^{\perp}(V)$ such that $QX^aQ^{\dagger}|\psi_i\rangle \in ker(V)$), then clearly there is no $\tilde{O}$, not even a state-specific one, for which $\tilde{O}V|\psi_i\rangle =V QX^a Q^{\dagger}|\psi_i\rangle$.
\end{remark}

\subsection{Proof of Theorem~\ref{thm:dvbound}}
\label{app:proof_thm2.2}
\begin{proposition}
Let $\chi(U)=U\chi_0$ be a 1-design where $\chi_0$ is defined in Theorem~\ref{thm:d1X0}, and  $V:A\rightarrow B$ is a linear map with singular values $\sigma_1,\sigma_2,\dots$. Define
\begin{align}
f_U(\{\sigma_k\},|A|,|B|)&:=2|A|\sum_{k=1}^{|B|}\sigma_k^2-\frac{2}{|B|}(\sum_{k=1}^{|B|}\sigma_k)^2\\
f_L(\{\sigma_k\},|A|,|B|)&:=\max\{0,2\Big[|A|(\sum_{k=1}^{|B|}\sigma_k^2)-\sqrt{|A||B|}(\sum_{k=1}^{|B|} \sigma_k)^2\Big]\} 
\end{align} 
then $f_L\leq  d_1(V,\chi(U))\leq f_U$ for any $U$.
\end{proposition}
\begin{proof}
From (\ref{DvOO}), we see that 
\begin{equation}
    D_V(\tilde{O},O') = \frac{2}{|A|}\Tr[V^{\dagger}V] - \frac{2}{|A|}Re\{\Tr[V^{\dagger}\tilde{O}VO'^{\dagger}]\}
\end{equation}
Now consider $UO'^{\dagger}\in U\chi_0$. Minimizing $D_V$ means maximizing \begin{equation}
    Re\{\Tr[V^{\dagger}\tilde{O}VO'^{\dagger}]\} = Re\{\sum_i\sigma_i\langle\tilde{\psi}_i|\tilde{O}VUO'^{\dagger}|\psi_i\rangle\}
    \label{eqn:realpart}
\end{equation}
Recall that Pauli operators $Z|j\rangle =\exp(i\phi_j)|j\rangle$ where $\phi_j = j \omega$ for some phase. Let $O'^{\dagger}=QX^aZ^{b}Q^{\dagger}$.
The quantity inside the curly bracket now reads
\begin{align}
&\sum_j \sigma_j \exp(i\phi_j b)\langle \tilde{\psi}_j| \tilde{O} VU  |\psi_{\mathscr{S}_a(j)}\rangle  \\
=&\sum_j\sigma_j\exp(i\phi_j b)\langle\tilde{\psi}_j|\tilde{O} V\sum_{k} U_{k\mathscr{S}_a(j)}|\psi_k\rangle\\
=&\sum_{j,k} \exp(i\phi_j b)\sigma_j\sigma_{k} \langle\tilde{\psi}_j|\tilde{O}|\tilde{\psi}_k\rangle U_{k\mathscr{S}_a(j)}\\
=&\sum_{j,k}\exp(i\phi_j b)\sigma_j\sigma_k \tilde{O}_{jk}U_{k,j \oplus a}
\end{align}
where $\mathscr{S}_a(j)=j\oplus a$ because it comes from $X^a$ acting on the computational basis states $|j\rangle \rightarrow |j\oplus a\rangle$. Again $\oplus$ is addition modulo $|A|$.

Note that $\exp(i\phi_j b)\sigma_j\sigma_k \tilde{O}_{jk}$ can be written as 
\begin{equation}
    [\tilde{M}]_{ij}=[\tilde{D}\Phi \tilde{O} \tilde{D}]_{ij}
\end{equation}
where $\tilde{D}=diag(\sigma_1,\sigma_2,\dots,\sigma_{|B|})$ and $\Phi=diag(\exp(i\phi_1 b),\exp(i\phi_2 b),\dots)$. Without loss of generality and for convenience, we order the indices such that $\sigma_1\geq \sigma_2\geq \dots$.

Minimizing $D_V$ over all $\tilde{O}$ corresponds to maximizing (\ref{eqn:realpart}). Since $\Phi$ is clearly a unitary and we will be varying $\tilde{O}$ anyway, we can absorb it into $\tilde{O}$ when we maximize $\Re\{\dots\}$ overall $\tilde{O}$. $\tilde{D}$ on the other hand is not a unitary in general.

 Simplifying the only term that has $\tilde{O}$ dependence
\begin{align}
    \sum_{O\in \chi_0}\max_{\tilde{O}} \sum_{ij} [\tilde{D}\tilde{O}\tilde{D}]_{ij} U_{j,i\oplus a} &=\sum_{a,b=0}^{|A|-1}\max_{\tilde{O}} \sum_{ij} [\tilde{D}\tilde{O}\tilde{D}]_{ij} U_{j,i\oplus a} \\
    &= |A|\sum_{a=0}^{|A|-1}\max_{\tilde{O}}\sum_{i,j=1}^{|B|} [\tilde{D}\tilde{O}\tilde{D}]_{ij} \tilde{R}(a)_{j,i} 
\end{align}
In the last equality, we have truncated all singular values with indices $>|B|$ in $\tilde{D}$ because they are zeros. Let $R_{i\leq |B|,j\leq |A|}$ be the first $|B|$ rows of $U$. Because of vanishing singular values $\sigma_{i>|B|}$ the rest of $U$ never contributes to the sum. We have also defined $\tilde{R}(a)$ to be a $|B|\times |B|$ matrix that one gets from $R$ by keeping only the $l=k+a$ to $l=k+a+|B|$th column where it is understood the indices are labelled as $(l-1)mod(|A|)+1$ thanks to the index permutation generated by $X^a$.

Applying lemma \ref{lemma:upperbound}, we get that $$S=\sum_{O\in \chi_0}\max_{\tilde{O}}\Re\{\Tr[V^{\dagger}\tilde{O}VUO'^{\dagger}]\}\leq |A|\sqrt{|A|}(\sum_{k=1}^{|B|} \sigma_k)^2$$
From the non-negativity of $d_1(V,U\chi_0)$, we must also have that $$\sum_{O\in \chi_0}\Tr[V^{\dagger}V]=|A|^2\sum_k(\sigma_k^2)-S\geq 0.$$ Hence a tighter upper bound for $S$ reads
\begin{equation}
    S\leq \min\{|A|\sqrt{|A|}(\sum_{k=1}^{|B|} \sigma_k)^2, |A|^2(\sum_{k=1}^{|B|}\sigma_k^2)\}
\end{equation}
Putting things together, \begin{equation}
   f_L(\{\sigma_k\},|A|,|B|)=\max\{0,2\Big[|A|(\sum_{k=1}^{|B|}\sigma_k^2)-\sqrt{|A|}(\sum_{k=1}^{|B|} \sigma_k)^2\Big]\} \leq d_1(V,U\chi_0)
\end{equation}
On the other hand, we apply the lower bound of $S$
\begin{equation}
    \frac{|A|}{|B|}(\sum_{k=1}^{|B|}\sigma_k)^2\leq S
\end{equation}
meaning that 
\begin{equation}
    d_1(V,\chi(U))\leq 2|A|\sum_{k=1}^{|B|}\sigma_k^2-\frac{2}{|B|}(\sum_{k=1}^{|B|}\sigma_k)^2 = f_U(\{\sigma_k\},|A|,|B|)
\end{equation}

%For a brief remark, let $||v||_1=\sum_k\lambda_k$, then the upper bound can be written as $2|A|(|A|||v||_2^2-\frac{1}{|B|}||v||_1^2)$. Since $||v||_2\geq ||v||_1/\sqrt{|B|}$, $$|A|||v||_2^2-\frac{1}{|B|}||v||_1^2\geq \frac{|A|-1}{|B|}||v||_1^2\geq 0.$$

%This bound, however, is far from tight again as the conjectured bound should not have the extra $1/|B|$ factor in the second term. If that were the case, we would have 
%$$f_U =|A|||v||_2^2-||v||_1^2\geq \frac{|A|-|B|}{|B|}||v||_1^2.$$
\end{proof}

This bound is non-trivial for $|A|\gg |B|$. It's easier to see it in a simpler instance in Sec.~\ref{app:proof_thm2.3}.

\begin{theorem}
  Let $\chi(U) = \{UW: W\in \chi_0\}$ for some unitary $U$. Then
\begin{align}
    &f_L(\{\sigma_k\},|A|,|B|)\leq d_\nu(V)\leq f_U((\{\sigma_k\},|A|,|B|).
\end{align}
\end{theorem}
\begin{proof}
%A heuristic proof only. 
If we can interchange the sum and the integral, then we can arrive at a lower bound for the intended state-dependence measure integrated over Haar. 
\begin{align}
    f_L(\{\sigma_k\},|A|,|B|)&\leq \frac{1}{vol(U_{|A|})}\int d\nu(U) d_1(V,\chi(U)) \\
    &= \frac{1}{vol(U_{|A|})}\int d\nu(U) \frac{1}{|\chi(U)|}\sum_{{O}\in \chi(U)} \min_{\tilde{O}} D_V(\tilde{O},O)\\
    &= \frac{1}{|\chi_0|}\sum_{W\in \chi_0}\int \frac{d\nu(U)}{vol(U_{|A|})} \min_{\tilde{O}} D_V(\tilde{O}, UW) \\
    &= \frac{1}{|\chi_0|} \sum_{W\in \chi_0}d_\nu(V) = d_\nu(V).
\end{align}
The top expression is by substituting the lower bound for any design $\chi(U)$ then integrate over a constant. The third equality assumes I can interchange sum and integral. $\chi(U)$ are isomorphic to $\chi_0$ so the volume does not change. Finally, integrating the Haar measure doesn't care whether we shift the variable by a constant unitary $W$, so they should all return the same outcome $d_\nu(V)$.

The same argument can be repeated to get the upper bound.
    
\end{proof}

\subsection{Proof of Theorem~\ref{thm:coiso_bound}}
\label{app:proof_thm2.3}

From Theorem~\ref{thm:dvbound} we can conclude that
\begin{corollary}\label{coro:coiso_bound}
    For $V$ proportional to co-isometries such that $\sigma_k=\eta$ for all $k=1,\dots, |B|$, then $$2\eta^2(|A||B|-|B|)\geq d_1(V,\chi(U))\geq \max\{0,2\eta^2(|A||B|-|A|^{1/2}|B|^{2})\}.$$
\end{corollary}
Note that the lower bound, though non-trivial for $|A|\geq |B|^2$, is not tight as we know that $d_1(V,\chi(U))>0$ as soon as $|A|>|B|$ from our definition. A more reasonable guess would be to put $|B|^{3/2}$ instead of $|B|^{2}$ in the second term as suggested by power counting and supported by small scale numerics. %However I had trouble proving it. It may be coming from the fact that when $||B||_1\geq ||B||_*$ is saturated when the rows vectors have a single 1. 
The upper bound for $||R||_{E_1}$ is far from tight in this case, where we should have $|B|$ instead of $|B|^{3/2}$. On the other hand, when $A$ is a normalized Hadamard matrix, then $||R||_{E_1}\leq |B|^{3/2}$ is tight, but $||B||_{E_1}\geq ||B||_*$ is not tight. Indeed, $||B||_{E_1}$ would yield $|B|^{3/2}$ whereas $||B||_*$ is only $|B|$. At least in these somewhat special limits, we see that there's an additional $\sqrt{|B|}$ factor that can be tightened.

%I conjecture the upper bound to be $2(|A|^2|B|-|A||B|^2)$ which differs from the current one by a factor of $|B|$ in the second term. It also coincides with the value of $d_1(V,\chi_0)$ exactly for co-isometries. Numerically I have not yet found counterexamples to this conjecture.

%It warrnts to repeat the above more carefully when I have time. As a corollary, this result also extends to the case where $V$ is proportional to co-isometry as we just need to pull out an overall constant from $V$ and make use of the above proof. 

For equal non-zero singular values, one can derive a tighter upper bound.
\begin{proposition}\label{prop:coiso_upperbound}
Suppose $V:A\rightarrow B$ is a co-isometry, i.e., all non-zero singular values are 1s, then
    \begin{equation*}
d_1(V,\chi(U))\leq 2|A|^2|B|-2|A||B|^2=d_1(V,\chi_0)
\end{equation*}

\end{proposition}

\begin{proof}
Recall that the nuclear norm $||B||_*=\sum_{i}\sigma_i\geq \sqrt{\sum_i\sigma_i^2}=||B||_F$. However, here we have something better because the singular values of the rectangular matrix $R$ are $1$s because its rows are taken from a unitary matrix. Since one obtains $B(k)$ from $R$ by removing columns, it must follow that the singular values of $B(k)$ obey $\sigma_i\leq 1$. This means that $\sum_i \sigma_i\geq \sum_i \sigma_i^2 =||B(k)||_F^2$ (without the square root).

We are interested in finding $\sum_k ||B(k)||_*$ where $B(k)$ are defined in the usual way by taking square submatrices from a rectangular matrix $R$.
Instead of applying the 1-norm lower bound, we use the tighter bound above: 
\begin{equation}
    \sum_{k=1}^{|A|} ||B(k)||_*\geq \sum_{k}||B(k)||_F^2=|B|||R||_F^2=|B|^2.
\end{equation}
The last equality follows because $R$ is made up of orthonormal row vectors whose 2-norms are 1. Adding $|B|$ rows yields the above. 

Then $\sum_{O\in U\chi}\max_{\tilde{O}}\Re\{\Tr[V^{\dagger}\tilde{O}VUO]\}\geq |A||B|^2$. Combining with the constant piece, we have that for co-isometries
\begin{equation}
d_1(V,\chi(U))\leq 2|A|^2|B|-2|A||B|^2
\end{equation}
%which is precisely the  conjectured upper bound.
%
For maps $V$ that are proportional to co-isometries by some constant $\eta$, we simply pull out these factors at the beginning and get 
$d_1(V,\chi(U))\leq (2|A|^2|B|-2|A||B|^2)\eta^2$
\end{proof}
From this, it follows that for maps that are co-isometries or proportional to co-isometries, $d_1(V,\chi_0)$ indeed plays a special role as it saturates the upper bound of all $d_1(V,\chi(U))$ and hence that of the Haar averaged measure $d_\nu(V)$.

%One should be able to generalize the above proof to unequal eigenvalues where instead of $|B|^2$ we would have $\sum_{ij}\lambda_i^2\lambda_j^2/\lambda_{\rm max}^2$ where $\lambda_{\rm max}$ is again the largest singular value of $V$. The tricky part is how the singular values of $B$ transform under conjugation by $D$ as our proof requires that the singular values to be bounded by 1, we need to rescale by the right amount. A somewhat trivial upper bound is of course $\sigma_i(DBD)\leq \sigma_i(B)\lambda_{\rm max}^2$ as $D$ is diagonal with max singular value $\lambda_{\rm max}$.

Combining Corollary~\ref{coro:coiso_bound} and Proposition~\ref{prop:coiso_upperbound}, we obtain Theorem~\ref{thm:coiso_bound}.

\bibliographystyle{JHEP}
\bibliography{references}
\end{document}